\documentclass[a4paper,UKenglish,cleveref, autoref, thm-restate,numberwithinsect]{lipics-v2021}

\title{Minimal and Canonical Quotients for Simulation Equivalences}

\titlerunning{Quotients for Simulation Equivalences}

\author{Eduardo Costa Martins}{Eindhoven University of Technology, The Netherlands}{e.j.costa.martins@tue.nl}{https://orcid.org/0009-0006-1706-896X}{}

\author{Tim A.\,C. Willemse}{Eindhoven University of Technology, The Netherlands}{t.a.c.willemse@tue.nl}{https://orcid.org/0000-0003-3049-7962}{}

\authorrunning{E. Costa Martins and T.\,A.\,C. Willemse}

\Copyright{Eduardo Costa Martins and Tim A.\,C. Willemse}



\category{} 

\relatedversion{} 

\funding{The Cynergy4MIE project is supported by the Chips Joint Undertaking and its members, including the top-up funding by National Authorities under Grant Agreement No 101140226.}

\nolinenumbers 

\hideLIPIcs  

\usepackage[noend]{algpseudocode}
\usepackage{algorithm}
\usepackage{appendix}
\usepackage{extra}

\begin{document}

\maketitle

\begin{abstract}
	Quotients have only been studied for a handful of equivalences in the linear time-branching time spectrum, for which there are results pertaining to canonicity and minimality. We extend these results to weak simulation equivalence and coupled similarity, two closely related equivalences induced by simulation preorders. We describe abstract procedures for transforming an LTS into a unique representative of its equivalence class, and for transforming an LTS into an equivalent state- and transition-minimal LTS. Moreover, we show the minimisation problem is NP-complete.
\end{abstract}

\section{Introduction}\label{sec:introduction}

\newcommand{\new}[1]{#1}

\new{
A Labelled Transition System is a mathematical model for capturing how a system performs actions to evolve from one configuration to another.
LTSs are typically used as a target model for higher-level languages such as Process Algebras or Petri nets, providing these languages with a formal semantic model.
Concepts such as behavioural equivalence are often studied at the level of LTSs.
For instance, in~\cite{van_glabbeek_linear_1993}, van Glabbeek introduced the linear time-branching time spectrum of behavioural equivalences for systems with unobservable actions ($\tau$-steps).

Given an equivalence relation on LTSs, one common transformation is to reduce the state space whilst preserving semantic equivalence.
A common reduction technique is quotienting, which involves the merging of states into equivalence classes, and adding transitions between equivalence classes according to a certain schema.
Quotienting is interesting for a variety of reasons.
It can be used in compositional and modular reasoning, thereby helping to combat the infamous state space explosion problem; it also offers a more abstract view on a system's behaviour by removing distinctions in the LTS that are not semantically important.

Of the behavioural equivalences studied in~\cite{van_glabbeek_linear_1993}, quotienting functions have been found for (divergence preserving) weak bisimilarity~\cite{bolognesi_fundamental_1987}, (divergence preserving) branching bisimilarity~\cite{eloranta_essential_1997}, and trace equivalence, where the latter is a well-known result from automata theory. 
There are also results related to quotienting for behavioural equivalences that do not take $\tau$-steps into account, such as bisimilarity and strong simulation equivalence~\cite{reniers_results_2014,simulation}.

It seems natural to assume that quotienting works for all the behavioural equivalences, but this is not known, nor is there---to the best of our knowledge---any argument for an underlying reason as to why it should.
For instance, even though quotienting functions exist for the most common behavioural equivalences, the quotients themselves differ in a number of ways. 
For bisimilarity, the quotient is \emph{canonical} (i.e.\ unique up to isomorphism), and \emph{minimal} in the number of states and transitions. 
For branching bisimilarity, the same holds if $\tau$-loops are removed afterwards~\cite{eloranta_essential_1997}. 
For weak bisimilarity, so-called duplicate transitions need to be removed as well~\cite{eloranta_essential_1997}. 
For these equivalences:
\begin{itemize}
    \item the minimal LTS is unique;
    \item it can be computed in polynomial time~\cite{jansen_om_2020,kanellakis_ccs_1990}; and
    \item the quotient is state minimal if unreachable states are removed.
\end{itemize}
Yet, for trace equivalence, this is not the case. 
Whilst there is a unique and minimal \emph{deterministic} LTS for a given equivalence class, there may be multiple distinct and minimal LTSs if non-determinism is permitted. 
Additionally, the quotient LTS may not be state minimal, and the computation of a minimal LTS is PSPACE-complete. 
Thus, there are few recurring patterns amongst quotients for different equivalences.

In this paper, we study the canonicity and minimality for quotienting for coupled simulation equivalence~\cite{coupledsimoriginal} and weak simulation equivalence~\cite{park_concurrency_1981,van_glabbeek_linear_1993}. 
First, we generalise the notion of duplicate transitions from the setting of weak bisimilarity.
Then, we introduce transformations based on the $\tau$-laws for weak simulation equivalence and for coupled similarity, namely $\tau.(x+y)=x+y$ and $\tau.(\tau.x+y)=\tau.x+y$.
Applying the $\tau$-laws from left to right is deterministic, and provides the ingredients for obtaining a canonical quotient.
Applying them from right to left is non-deterministic, and provides the means for a minimal quotient.

When applying the $\tau$-laws from right to left, there is an optimal choice that provides the largest reduction in size.
However, finding a minimal quotient can be as hard as finding such an optimal choice, which itself is as hard as finding a minimal set cover, which is NP-complete, and minimisation turns out to be NP-complete, too.
The procedure we give for computing a minimal quotient relies on some solver for the set cover problem.
This way, we can leverage state of the art algorithms for the set cover problem in order to minimise state spaces.
Moreover, the instances of the set cover problem that need to be solved are typically small.\smallskip

\noindent\textit{Related Work.} In~\cite{bolognesi_fundamental_1987} it is shown how LTSs can be quotiented w.r.t.\ weak bisimilarity by solving the relational coarsest partition problem. It was shown in~\cite{eloranta_minimizing_1991} that by removing so-called \emph{redundant} transitions from a quotiented LTS, the resulting LTS is unique and minimal w.r.t.\ weak bisimilarity. The results from~\cite{eloranta_minimizing_1991} were generalised to branching bisimilarity and their divergence-preserving variants in~\cite{eloranta_essential_1997}. A canonical and minimal quotient for simulation equivalence (without internal $\tau$-steps) in the context of Kripke structures was found in~\cite{simulation}. The results from~\cite{simulation} extend to LTSs using translations found in~\cite{reniers_results_2014}.\smallskip

\noindent\textit{Outline.} 
\Cref{sec:prelim} contains preliminaries. In \Cref{sec:quotient} we show that the $\forall$-quotient is sound for both equivalences, though not necessarily state- or transition-minimal, nor canonical. In \Cref{sec:can}, we describe a procedure for a canonical quotient, which is not necessarily minimal. In \Cref{sec:min} we describe a procedure for a minimal quotient. This procedure can be modified so that the output is also canonical. Lastly, in \Cref{sec:conc} we conclude our findings. Proofs can be found in the appendix.
}

\newcommand{\old}[1]{}

\old{
An LTS is a semantic model that captures how a system performs actions to evolve from one configuration to another. It is typical to reduce the state space of these models whilst preserving semantic equivalence. A common reduction technique is quotienting, which involves the merging of states into equivalence classes, and adding transitions between equivalence classes according to a certain schema. In~\cite{van_glabbeek_linear_1993}, van Glabbeek introduced the linear time-branching time spectrum of behavioural equivalences for systems with $\tau$-steps. Of these behavioural equivalences, quotienting functions have been found for (divergence preserving) weak bisimilarity~\cite{bolognesi_fundamental_1987}, (divergence preserving) branching bisimilarity~\cite{eloranta_essential_1997}, and trace equivalence, where the latter is a well-known result from automata theory. There are also results related to quotienting for behavioural equivalences that do not take $\tau$-steps into account, such as bisimilarity and strong simulation equivalence~\cite{reniers_results_2014,simulation}. It seems natural to assume that quotienting works for all the behavioural equivalences, but this is not known, nor is there---to the best of our knowledge---any argument for an underlying reason as to why it should.

Moreover, even though quotienting functions exist for the most common behavioural equivalences, the quotients themselves differ in a number of ways. For bisimilarity, the quotient is unique up to isomorphism, and minimal in the number of states and transitions. For branching bisimilarity, the same holds if $\tau$-loops are removed afterwards~\cite{eloranta_essential_1997}. For weak bisimilarity, so-called duplicate transitions need to be removed as well~\cite{eloranta_essential_1997}. For these equivalences:
\begin{itemize}
    \item the minimal LTS is unique;
    \item it can be computed in polynomial time~\cite{jansen_om_2020,kanellakis_ccs_1990}; and
    \item the quotient is state minimal if unreachable states are removed.
\end{itemize}
Yet, for trace equivalence, this is not the case. Whilst there is a unique and minimal \emph{deterministic} LTS for a given equivalence class, there may be multiple distinct and minimal LTSs if non-determinism is permitted. Additionally, the quotient LTS may not be state minimal, and the computation of a minimal LTS is NP-hard. Thus, there are few recurring patterns amongst quotients for different equivalences.

In this paper, we work towards generalising results about quotienting to other behavioural equivalences. We specifically look at the coupled simulation and weak simulation preorders, and their induced equivalences. First, we generalise the notion of duplicate transitions from the setting of weak bisimilarity. Then, we introduce transformations based on the $\tau$-laws for weak simulation equivalence and for coupled similarity, namely $\tau.(x+y)=x+y$ and $\tau.(\tau.x+y)=\tau.x+y$. Applying the $\tau$-laws from left to right is deterministic, and provides the means for a canonical quotient. Applying them from right to left is non-deterministic, and provides the means for a minimal quotient.

When applying the $\tau$-laws from right to left, there is an optimal choice that provides the largest reduction in size. Finding a minimal quotient can be as hard as finding such an optimal choice, which itself is as hard as finding a minimal set cover, and therefore minimisation is NP-complete. For this reason, the procedure we give for computing a minimal quotient relies on some solver for the set cover problem. This way, we can leverage state of the art algorithms for the set cover problem in order to minimise state spaces. Moreover, the instances of the set cover problem that need to be solved are typically small.

\textsf{\bfseries Related Work.} In~\cite{bolognesi_fundamental_1987} it is shown how LTSs can be quotiented w.r.t. weak bisimilarity by solving the relational coarsest partition problem. It was shown in~\cite{eloranta_minimizing_1991} that by removing so-called \emph{redundant} transitions from a quotiented LTS, the resulting LTS is unique and minimal w.r.t. weak bisimilarity. The results from~\cite{eloranta_minimizing_1991} were generalised to branching bisimilarity and their divergence-preserving variants in~\cite{eloranta_essential_1997}. A canonical and minimal quotient for simulation equivalence (without internal $\tau$-steps) in the context of Kripke structures was found in~\cite{simulation}. The results from~\cite{simulation} extend to LTSs using translations found in~\cite{reniers_results_2014}.

\textsf{\bfseries Outline.} \Cref{sec:prelim} contains preliminaries. In \Cref{sec:quotient} we show that the $\forall$-quotient is sound for both equivalences, though not necessarily state- or transition-minimal, nor canonical. In \Cref{sec:can}, we describe a procedure for a canonical quotient, which is not necessarily minimal. In \Cref{sec:min} we describe a procedure for a minimal quotient. This procedure can be modified so that the output is also canonical. Lastly, in \Cref{sec:conc} we conclude our findings. Proofs can be found in the appendix.
}

\section{Preliminaries}\label{sec:prelim}

We recall Labelled Transition Systems, discuss the two notions of equivalence considered in this paper, and introduce the notational conventions we use throughout the paper.

\begin{definition}[Labelled Transition System]
    A \emph{labelled transition system} (LTS) is a 4-tuple $L=\langle S, \act_\tau, \rightarrow, \iota\rangle$ with:
    \begin{itemize}
        \item $S$ a set of states;
        \item $\act_\tau$ a set of actions of which $\tau$, the unobservable action, is an element;
        \item ${\rightarrow}\subseteq S\times\mathcal{A}_\tau\times S$ a transition relation; and
        \item $\iota\in S$ is the initial state.
    \end{itemize}
\end{definition}

Henceforth, we only consider transition systems with a finite number of states and transitions. We fix the set of actions $\act_\tau$, and write $\act$ to denote $\act_\tau\setminus\set{\tau}$. Elements of $\act$ are called \emph{visible}, whereas $\tau$ is \emph{invisible}. We use the following notation:
\begin{itemize}
    \item $p\xrightarrow{\alpha}q$ denotes $(p,\alpha,q)\in{\rightarrow}$,
    \item $p\xrightarrow{X}q$, for $X\subseteq\act_\tau$, denotes $p\xrightarrow{\alpha}q$ for each $\alpha\in X$,
    \item $p\xrightarrow{\alpha}$ means there exists some $q\in S$ such that $p\xrightarrow{\alpha}q$,
    \item $p\nxrightarrow{\alpha}q$ denotes $(p,\alpha,q)\notin{\rightarrow}$.
\end{itemize}
Moreover, for a transition $p\xrightarrow{\alpha}q$, we call $p$ and $q$ the \emph{source} and \emph{target} states of the transition, respectively.

For an LTS, $L=\langle S, \act_\tau, \rightarrow, \iota\rangle$, we may remove or add transitions as follows. Let $Y\subseteq S\times \act_\tau\times S$. Then the LTSs $L \cup Y$ and $L \setminus Y$ are defined as follows:
\begin{itemize}
    \item $L\cup Y\coloneqq\langle S, \act_\tau, {\rightarrow}\cup Y, \iota\rangle$; and similarly
    \item $L\setminus Y\coloneqq\langle S, \act_\tau, {\rightarrow}\setminus Y, \iota\rangle$.
\end{itemize}
We frequently consider LTSs with some states in common. Say there are LTSs $L, \bar{L}, L_i$ that have some states, $p$, $q$, in common. To distinguish transitions in each LTS, we use the following notational conventions:
\begin{itemize}
	\item $p\xrightarrow{\alpha}q$ as usual means `$(p,\alpha,q)$ is a transition of $L$,'
	\item $\bar{p}\xrightarrow{\alpha}\bar{q}$ is shorthand for `$(p,\alpha,q)$ is a transition of $\bar{L}$,' and
	\item $p_i\xrightarrow{\alpha}q_i$ is shorthand for `$(p,\alpha,q)$ is a transition of $L_i$.'
\end{itemize}
Here, we assume that the names $\bar{p}, p_i, \bar{q}, q_i$ are not part of the LTSs.

The main behavioural equivalences we discuss in this paper are weak simulation equivalence and coupled similarity. Whilst not the main focus, weak bisimilarity is relevant as well. These equivalences are defined using \emph{weak} transitions:
\begin{itemize}
    \item $p\xRightarrow{\tau}q$ iff $q$ is reachable from $p$ using zero or more $\tau$-transitions, and
    \item $p\xRightarrow{a}q$ iff $\exists p',q':p\xRightarrow{\tau}p'\xrightarrow{a}q'\xRightarrow{\tau}q$.
\end{itemize}
On the other hand, transitions of the form $p\xrightarrow{\alpha}q$ with $\alpha\in\act_\tau$ are called \emph{strong} transitions.

\begin{definition}[Weak Simulation]
    Let $L=\langle S, \act_\tau, \rightarrow, \iota\rangle$ be an LTS. A binary relation $\RR$ on $S$ is a weak \emph{simulation} on $L$ iff for all $(p,q)\in\RR,\alpha\in\act_\tau,p'\in S$, the \emph{simulation condition} holds:
    \[
        p\xrightarrow{\alpha}p'\implies\exists q'\in S: q\xRightarrow{\alpha}q'\land p'\RR q'\quad.
    \]
    We write $p\sqsubseteq_S q$ iff there exists a weak simulation $\RR$ such that $p\RR q$, and say that $q$ \emph{simulates}~$p$.

	Moreover, a symmetric weak simulation is called a \emph{weak bisimulation.} We write $p\bisim_w q$ iff there exists a weak bisimulation $\RR$ such that $p\RR q$, and say that $p$ and $q$ are \emph{weakly bisimilar}.
\end{definition}

    

\begin{definition}[Coupled Simulation]
    Let $L=\langle S, \act_\tau, \rightarrow, \iota\rangle$ be an LTS. A relation $\RR$ on $S$ is a \emph{coupled simulation} on $L$ iff it is a weak simulation, and moreover for all $(p,q)\in\RR$, the \emph{coupling condition} holds:
    \[
        \exists q'\in S: q\xRightarrow{\tau}q'\land q'\RR p\quad.
    \]
    We write $p\sqsubseteq_{CS} q$ iff there exists a coupled simulation $\RR$ such that $p\RR q$, and say that $q$ \emph{coupled simulates} $p$.
\end{definition}

Throughout this paper, we use the fact that $\sqsubseteq_S$ is the greatest weak simulation, and $\sqsubseteq_{CS}$ is the greatest coupled simulation~\cite{Coupledsim_Contrasim-AFP}. 

In this paper, we provide a number of results that apply to both preorders and their induced equivalences. To avoid repetition, we fix $x\in\set{S,CS}$ so that $\sqsubseteq_x$ describes the weak simulation preorder, or the coupled similarity preorder. An $x$-simulation is a weak simulation if $x=S$, and a coupled simulation if $x=CS$. For states $p,q$, and LTSs $L$ and $L'$ with initial states $\iota$ and $\iota'$, respectively, we write:
\begin{itemize}
    \item $p\equiv_x q$ iff $p\sqsubseteq_x q$ and $q\sqsubseteq_x p$,
    \item $p\sqsubset_x q$ iff $p\sqsubseteq_x q$, but not $q\sqsubseteq_x p$,
    \item $p\sqsupseteq_x q$ and $p\sqsupset_x q$ iff $q\sqsubseteq_x p$ and $q\sqsubset_x p$, respectively,
    \item $L\equiv_x L'$ iff $\iota\equiv_x\iota'$ in the disjoint union of $L$ and $L'$,
    \item $[p]_x$ for the $\equiv_x$ equivalence class of $p$, and
    \item $p\sqsupset_x$ to denote the set $\setof{q}{p\sqsupset_x q}$.
\end{itemize}
When we say equivalent, we mean equivalent with respect to $\equiv_x$. The examples in this paper work for both $\equiv_S$ (weak simulation equivalence) and $\equiv_{CS}$ (coupled similarity) unless stated otherwise. An \emph{inert} transition is one of the form $p\xrightarrow{\tau}q$ such that $p$ and $q$ are equivalent. 

Quotienting describes the merging of equivalent states in an LTS in such a way that the resulting LTS is again equivalent to the starting LTS.

\begin{definition}[Quotient]
	An LTS, $L=\langle S, \act_\tau, \rightarrow, \iota\rangle$, is a \emph{quotient} LTS iff for all states $p,q\in S$, $p\equiv_x q$ implies $p=q$. 
	A function $\QQ$ mapping one LTS to another is a \emph{quotienting function} iff for all LTSs $L$, $\QQ(L)\equiv_x L$ and $\QQ(L)$ is a quotient LTS. We say that $\QQ(L)$ is a quotient of $L$.
\end{definition}

We use isomorphism to describe uniqueness of a quotient.


\begin{definition}[Isomorphism]
    Let $L=\langle S, \act_\tau, \rightarrow, \iota\rangle$ and $L'=\langle S', \act_\tau, \rightarrow', \iota'\rangle$ be two LTSs. We say $L$ and $L'$ are \emph{isomorphic}, written $L\sim L'$, iff there is a bijection $f:S\rightarrow S'$ such that:
    \begin{itemize}
        \item $f(\iota)=\iota'$; and
        \item $(p,\alpha,q)\in{\rightarrow}$ iff $(f(p), \alpha, f(q))\in{\rightarrow}'$.
    \end{itemize}
    Let $\QQ$ be some quotienting function and $L$ some LTS. We say $\QQ(L)$ is \emph{canonical} or \emph{unique} iff $\QQ(L)\sim\QQ(L')$ for all $L'\equiv_x L$.
\end{definition}

As is typical when discussing quotients~\cite{simulation,eloranta_essential_1997}, we define the size of an LTS using a lexicographic ordering on states and transitions, thereby prioritising a reduction in the number of states over a reduction in the number of transitions.

\begin{definition}[Size]
    Let $L=\langle S, \act_\tau, \rightarrow, \iota\rangle$ and $L'=\langle S', \act_\tau, \rightarrow', \iota'\rangle$ be two LTSs. We write $|L|\leq|L'|$ iff:
    \begin{itemize}
        \item $|S|<|S'|$; or
        \item $|S|=|S'|$ and $|{\rightarrow}|\leq|{\rightarrow}'|$.
    \end{itemize}
    We say $L$ is \emph{minimal} (for $\equiv_x$) iff $|L|\leq |L'|$ for all $L'\equiv_x L$.
\end{definition}

\section{Universal Quotient and Redundant Transitions}\label{sec:quotient}

The universal quotient, or $\forall$-quotient, is a common transformation~\cite{simulation} for various behavioural equivalences involving the merging of equivalent states, see \Cref{fig:forall_quotient}. It gets its name from the $\forall$-quantifier that appears in the definition of the transition relation.

\begin{definition}[$\forall$-Quotient]
    Let $L=\langle S, \act_\tau, \rightarrow, \iota\rangle$ be an LTS. The $\forall$\emph{-quotient} of $L$ is an LTS, $L_{/\equiv}=\langle S_{/\equiv}, \act_\tau, \rightarrow_{/\equiv}, [\iota]_x\rangle$, where:
    \begin{itemize}
        \item $S_{/\equiv}=\setof{[p]_x}{p\in S}$; and
        \item $\rightarrow_{/\equiv}=\setof{
        ([p]_x,\alpha,[q]_x)}{
        \forall p'\in[p]_x,\exists q'\in[q]_x:p'\xRightarrow{\alpha}q'\land (\alpha=\tau\implies p'\not\equiv_x q')
        }$.
    \end{itemize}
\end{definition}
\begin{remark}
The transition relation for the $\forall$-quotient is defined such that $L_{/\equiv}$ has no inert transitions. Successive transformations we define require the absence of such inert transitions, and also do not introduce inert transitions.
\end{remark}

\begin{figure}[H]
    \centering
    \begin{tikzpicture}[node distance = 1.5cm]
        \tikzset{every path/.style={thick}}
        \node[node] (p) {$p$};
        \node[node] (q)[above right = 0.75cm and 1.5cm of p] {$q$};
        \node[node] (r)[below right = 0.75cm and 1.5cm of p] {$r$};
        \node[node] (s)[right=of q] {$s$};
        \node[node] (sp)[right=of r] {$s'$};
        \coordinate[above left=0.25cm and 0.43cm of p] (initL);
        \draw[->] (initL) to (p);

        \draw[->] (p) to["$a$"] (q);
        \draw[->, loop above] (q) to["$b$"] (q);
        \draw[->] (q) to["$\tau$"] (s);
        \draw[->] (p) to["$a$" below left] (r);
        \draw[->] (r) to["$\tau$"] (s);
        \draw[->] (r) to["$b$" below] (sp);

        \coordinate[] (x) at ($(s)!0.5!(sp)$);
        \node[node] (pc) [right=2.5cm of x] {$[p]_x$};
        \node[node] (qc) [above right= 0.75cm and 1.5cm of pc] {$[q]_x$};
        \node[node] (rc) [below right= 0.75cm and 1.5cm of pc] {$[r]_x$};
        \node[node] (sc) [below right= 0.75cm and 1.5cm of qc] {$[s]_x$};
        \coordinate[above left=0.25cm and 0.43cm of pc] (initLc);
        \draw[->] (initLc) to (pc);
        \draw[|->, shorten <=0.75cm, shorten >=0.75cm] (x) to (pc);
        
        \draw[->] (pc) to["$a$"] (qc);
        \draw[->, loop above] (qc) to["$b$"] (qc);
        \draw[->] (qc) to["$\tau$" below left] (sc);
        \draw[->] (pc) to["$a$" below left] (rc);
        \draw[->, bend right] (rc) to["$b$" below right] (sc);
        \draw[->] (rc) to["$\tau$"] (sc);

        \draw[->] (pc) to["$a$"] (sc);
        \draw[->, bend left,] (qc) to["$b$"] (sc);
        
    \end{tikzpicture}
    \caption{The $\forall$-quotient maps the left LTS to the right LTS, where $s'\in[s]_x$.}
    \label{fig:forall_quotient}
\end{figure}
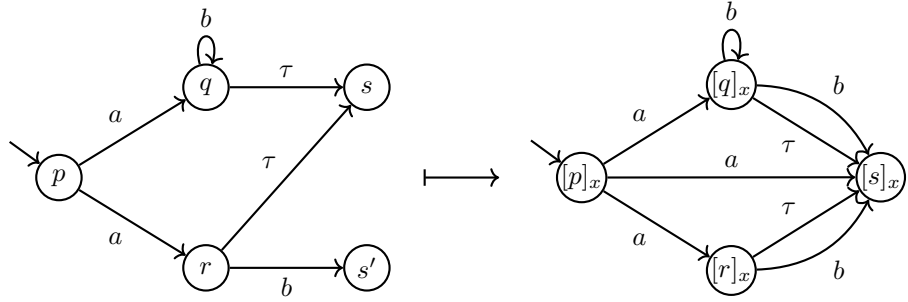

That the $\forall$-quotient \emph{is} indeed a quotienting function is stated in \Cref{prop:quotient}, the proof of which relies on the following lemmas. These require the following definitions, where $\alpha\in\act_\tau$ and $X$ is a set of states, in the context of an LTS $L=\langle S,\act_\tau,\rightarrow,\iota\rangle$.
\begin{align*}
	G_\alpha(p)&=\setof{p'}{p\xRightarrow{\alpha}p'}\\
    \Max(X)&=\setof{p\in X}{\forall q\in X: p\sqsubseteq_x q\implies q\sqsubseteq_x p}\\
    \Min(X)&=\setof{p\in X}{\forall q\in X: q\sqsubseteq_x p\implies p\sqsubseteq_x q}
\end{align*}
We remark that whenever $X$ is non-empty and finite, then $\Max(X)$ and $\Min(X)$ are non-empty.
\begin{restatable}{lemma}{tausucc}\label{lem:tau_succ}
    If $p\xRightarrow{\tau} q$, then $q\sqsubseteq_x p$.
\end{restatable}

\begin{restatable}{lemma}{simact}\label{lem:sim_act}
    If $p\sqsubseteq_x q$ and $p\xRightarrow{\alpha}p'$, then there exists $q'$ such that $q\xRightarrow{\alpha}q'$ and $p'\sqsubseteq_x q'$.	
\end{restatable}

\begin{restatable}{lemma}{lmax}\label{lem:max}
    Let $p\in S,\alpha\in\act_\tau$. Then, for all $q\equiv_x p$ and all $p'\in\Max(G_\alpha(p))$, there is some $q'\in S$ such that $q\xRightarrow{\alpha}q'$ and $p'\equiv_x q'$.
\end{restatable}

\begin{restatable}{lemma}{lmin}\label{lem:min}
    Let $p\in S$ and suppose $x=CS$. Then, for all $q\in S$ such that $p\sqsubseteq q$, and all $p'\in\Min(G_\tau(p))$, there is some $q'\in S$ such that $q\xRightarrow{\tau}q'$ and $p'\equiv_x q'$.	
\end{restatable}

\begin{restatable}[$\forall$-Quotient]{proposition}{quotient}\label{prop:quotient}
    The function mapping an LTS $L$ to its $\forall$-quotient, $L_{/\equiv}$, is a quotienting function. Moreover, any state $p$ of $L$ is equivalent to the state $[p]_x$ of $L_{/\equiv}$.
\end{restatable}

If an LTS has at least two equivalent states, then the $\forall$-quotient guarantees a reduction in the size of the LTS. However, if no two states are equivalent, then the $\forall$-quotient essentially computes the transitive closure of the given LTS. That is, it adds a $p\xrightarrow{\alpha}q$ transition whenever $p\xRightarrow{\alpha}q$, potentially creating a \emph{duplicate transition} and increasing the size of the LTS.

\begin{definition}[Duplicate Transitions~\cite{eloranta_essential_1997}]
	Let $L=\langle S, \act_\tau, \rightarrow, \iota\rangle$ be an LTS without inert transitions. We say a transition $p\xrightarrow{\alpha}q$ of $L$ \emph{duplicates} $p'\xrightarrow{\alpha}q'$ iff:
    \begin{itemize}
        \item $p\xRightarrow{\tau}p'$, 
		\item $q'\xRightarrow{\tau}q$, and
        \item $(p,q)\neq(p',q')$.
    \end{itemize}
    We say $p\xrightarrow{\alpha}q$ is a \emph{duplicate} if there is some $p'\xrightarrow{\alpha}q'$ that $p\xrightarrow{\alpha}q$ duplicates.
\end{definition}

Conversely, if some $p\xrightarrow{\alpha}q$ transition is a duplicate, then $p$ can reach $q$ with a weak $\alpha$-transition without using the duplicate $p\xrightarrow{\alpha}q$ transition, which can be removed. Consider the LTS, $L$, on the right in \Cref{fig:forall_quotient}, and the LTS, $L'$, on the left in \Cref{fig:covered_intro}. The two LTSs are equivalent, and computing the $\forall$-quotient of $L'$ results in (an LTS isomorphic to) the larger LTS $L$. Conversely, removing all duplicate transitions of $L$ results in (an LTS isomorphic to) the smaller LTS $L'$.

\begin{figure}[ht]
    \centering
    \begin{tikzpicture}[node distance = 1.5cm]
        \tikzset{every path/.style={thick}}
        \node[node] (p) {$p$};
        \node[node] (q)[above right = 0.75cm and 1.5cm of p] {$q$};
        \node[node] (r)[below right = 0.75cm and 1.5cm of p] {$r$};
        \node[node] (s)[above right=0.75cm and 1.5cm of r] {$s$};
        \coordinate[above left=0.25cm and 0.43cm of p] (initL);
        \draw[->] (initL) to (p);

        \draw[->] (p) to["$a$"] (q);
        \draw[->, loop above] (q) to["$b$"] (q);
        \draw[->, bend left] (q) to["$\tau$"] (s);
        \draw[->] (p) to["$a$"] (r);
        \draw[->, bend left] (r) to["$\tau$"] (s);
        \draw[->, bend right] (r) to["$b$" below right] (s);
        
        \node[node] (pp)[right=2.5cm of s] {$p$};
        \node[node] (qp)[above right = 0.75cm and 1.5cm of pp] {$q$};
        \node[node] (rp)[below right = 0.75cm and 1.5cm of pp] {$r$};
        \node[node] (sp)[above right=0.75cm and 1.5cm of rp] {$s$};
        \coordinate[above left=0.25cm and 0.43cm of pp] (initLp);
        \draw[->] (initLp) to (pp);

        \draw[->] (pp) to["$a$"] (qp);
        \draw[->, loop above] (qp) to["$b$"] (qp);
        \draw[->, bend left] (qp) to["$\tau$"] (sp);
        \draw[->, bend left] (rp) to["$\tau$"] (sp);
        \draw[->, bend right] (rp) to["$b$" below right] (sp);
        
        \draw[|->, shorten <=0.75cm, shorten >=0.75cm] (s) to (pp);
    \end{tikzpicture}
    \caption{Two equivalent LTSs illustrating the need for a stronger notion of duplicate transitions.}
    \label{fig:covered_intro}
\end{figure}
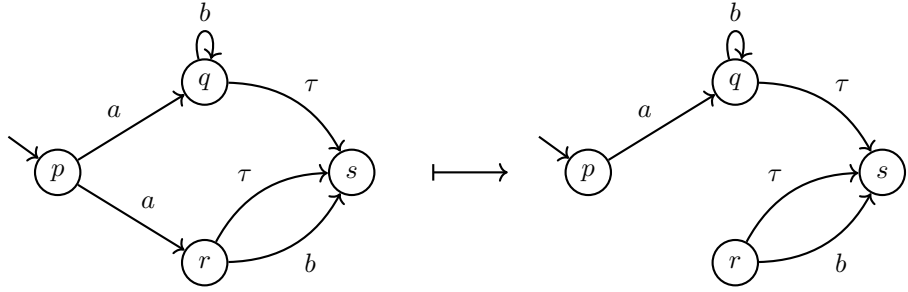

Removing all duplicate transitions after quotienting results in an LTS no larger than the input LTS---we prove a slightly different result in \Cref{prop:covered_quotient}. For weak bisimilarity, this procedure describes a canonical, minimal LTS~\cite{eloranta_minimizing_1991,eloranta_essential_1997}. However, this is not the case for coupled or weak simulation equivalence, as there may be removable transitions that are not included in the notion of duplicate transitions. Consider the LTSs in \Cref{fig:covered_intro}. The $p\xrightarrow{a}r$ transition may be removed since $p\xrightarrow{a}q$ and $r\sqsubset_x q$. Consequently, $r$ may be removed since it is unreachable. The $p\xrightarrow{a}r$ transition is what we call a \emph{covered transition}. Removing covered transitions is similar to removing duplicate transitions, except that we do not require a weak transition to the same state as the covered transition, but simply to some greater state, see \Cref{fig:covered_demo}. The concept of a covered transition resembles that of a `little brother' in the context of simulation equivalence for Kripke structures~\cite{DBLP:conf/icalp/KuceraM99,simulation}.

\begin{definition}[Covered Transitions]\label{def:covered}
    Let $L=\langle S, \act_\tau, \rightarrow, \iota\rangle$ be an LTS without inert transitions. We say a transition $p\xrightarrow{\alpha}q$ of $L$ is \emph{covered by} $p'\xrightarrow{\alpha}q'$ iff:
    \begin{itemize}
        \item $p\xRightarrow{\tau}p'$,
        \item $q\sqsubseteq q'$, and
        \item $(p,q)\neq(p',q')$.
    \end{itemize}
    We say $p\xrightarrow{\alpha}q$ is \emph{covered} if there is some $p'\xrightarrow{\alpha}q'$ that covers $p\xrightarrow{\alpha}q$.
\end{definition}

\begin{figure}[ht]
    \centering
	\hfill
    \begin{tikzpicture}[node distance = 1.5cm]
        \tikzset{every path/.style={thick}}
        \node[node] (p) {$p$};
        \node[node] (pp)[below=of p] {$p'$};
        \node[node] (rp)[right=of pp] {$q'$};
        \node[node] (q)[above=of rp] {$q$};

        \draw[-Implies, double equal sign distance] (p) to["$\tau$" left] (pp);
        \draw[->] (pp) to["$\alpha$" below] (rp);
        \draw[-Implies, double equal sign distance] (rp) to["$\tau$" right] (q);        
        \draw[->] (p) to["$\alpha$"] (q);
    \end{tikzpicture}
	\hfill
    \begin{tikzpicture}[node distance = 1.5cm]
        \tikzset{every path/.style={thick}}
        \node[node] (p) {$p$};
        \node[node] (pp)[below=of p] {$p'$};
        \node[node] (rp)[right=of pp] {$q'$};
        \node[node] (q)[above=of rp] {$q$};

        \draw[-Implies, double equal sign distance] (p) to["$\tau$" left] (pp);
        \draw[->] (pp) to["$\alpha$" below] (rp);
        \draw[->, dashed] (q) to (rp);        
        \draw[->] (p) to["$\alpha$"] (q);
    \end{tikzpicture}
	\hfill\mbox{}
    \caption{Comparing duplicate and covered transitions. The $p\xrightarrow{\alpha}q$ transition on the left (right) is duplicated (covered) by the $p'\xrightarrow{\alpha}q'$ transition.}
    \label{fig:covered_demo}
\end{figure}
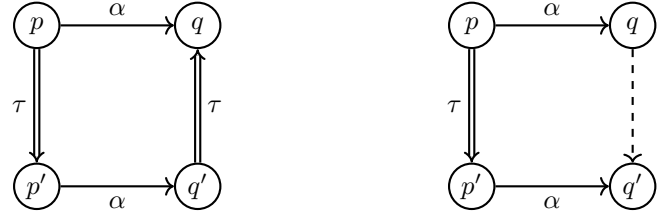

\begin{restatable}{lemma}{pcovered}\label{prop:covered} 
    Let $L=\langle S, \act_\tau, \rightarrow, \iota\rangle$ be an LTS without inert transitions, and $p\xrightarrow{\alpha}q$ some transition of $L$ covered by $p'\xrightarrow{\alpha}q'$. Let $s,t\in S$ be arbitrary, then:
    \[
        s\xRightarrow{\tau}p'\xrightarrow{\alpha}q'\xRightarrow{\tau}t\text{ in }L\iff s\xRightarrow{\tau}p'\xrightarrow{\alpha}q'\xRightarrow{\tau}t\text{ in }L\setminus\set{(p,\alpha,q)}\quad.
    \]	
\end{restatable}

\begin{restatable}[Removing Covered Transitions]{corollary}{ccovered}\label{cor:covered}
    Let $L=\langle S, \act_\tau, \rightarrow, \iota\rangle$ be an LTS without inert transitions, and let $p\xrightarrow{\alpha}q$ be some covered transition of $L$. Let $\bar{L}\coloneqq L\setminus\set{(p,\alpha,q)}$. Then $L\equiv_x \bar{L}$. Moreover, for all $s\in S$, it holds that $s$ in $L$ is equivalent to $s$ in $\bar{L}$.	
\end{restatable}

The following lemma states that a covered transition remains covered after the removal of some other covered transition. Therefore, one may identify all the covered transitions of some LTS, $L$, remove them all at once, and the resulting LTS is equivalent. Moreover, since removing a covered transition does not introduce new covered transitions, the resulting LTS has no covered transitions.

\begin{restatable}{lemma}{coveredpersist}\label{lem:covered_persist}
    Let $L=\langle S, \act_\tau, \rightarrow, \iota\rangle$ be an LTS without inert transitions. Let $p\xrightarrow{\alpha}q$ and $s\xrightarrow{\beta}t$ be two distinct covered transitions of $L$. Let $\bar{L}\coloneqq L\setminus\set{(p,\alpha,q)}$, then $s\xrightarrow{\beta}t$ is a covered transition of $\bar{L}$.	
\end{restatable}

By computing the $\forall$-quotient of an LTS, and then removing all covered transitions, we obtain a reduction.

\begin{restatable}{proposition}{coveredquotient}\label{prop:covered_quotient}
    Let $L=\langle S, \act_\tau, \rightarrow, \iota\rangle$ be an LTS. Let $\bar{L}=\langle \bar{S}, \act_\tau, \bar{\rightarrow}, \bar{\iota}\rangle$ be the LTS obtained from $L$ by computing its $\forall$-quotient and subsequently removing all covered transitions. Then $|\bar{L}|\leq|L|$.	
\end{restatable}

Surprisingly, even with a generalised notion of duplicate transitions and the subsequent removal of unreachable states, we do not have a canonical or minimal quotient. Consider the LTSs in \Cref{fig:not_can_min}. They are both equivalent, yet we cannot reduce the number of states of the left LTS using the transformations listed so far. Therefore, this procedure does not lead to a minimal quotient. Similarly, the transformations listed so far do not introduce new states, so we do not have a canonical quotient since this would require reducing the number of states of the left LTS, or increasing the number of states of the right LTS.

\begin{figure}[ht]
    \centering
    \begin{tikzpicture}[node distance = 1.5cm]
        \tikzset{every path/.style={thick}}
        \node[node] (p) {$p$};
        \node[node] (q)[above right=0.75cm and 1.5cm of p] {$q$};
        \node[node] (r)[below right=0.75cm and 1.5cm of p] {$r$};
        \coordinate[left=0.5cm of p] (initL);
        \draw[->] (initL) to (p);
        \coordinate[] (x) at ($(q)!0.5!(r)$);

        \draw[->, bend left] (p) to["$\tau$"] (q);
        \draw[->, bend right] (q) to["$\tau$" left] (r);
        \draw[->, bend left] (q) to["$b$"] (r);
        \draw[->, bend right] (p) to["$a$" below left] (r);
        
        \node[node] (pp)[right=2.8cm of x] {$p$};
        \node[node] (rp)[right=of pp] {$r$};
        \coordinate[left=0.5cm of pp] (initLp);
        \draw[->] (initLp) to (pp);

        \draw[->, bend left=45] (pp) to["$b$"] (rp);
        \draw[->] (pp) to["$\tau$"] (rp);
        \draw[->, bend right=45] (pp) to["$a$" below] (rp);
    \end{tikzpicture}
    \caption{The two LTSs are equivalent, but with the transformations introduced so far we cannot reduce the size of the left LTS, or increase the number of states with the right LTS.}
    \label{fig:not_can_min}
\end{figure}
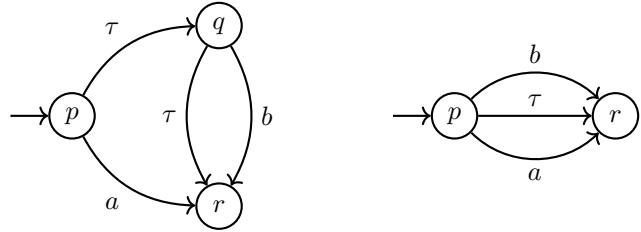

Moreover, the problem of finding a canonical quotient cannot necessarily be solved by finding a minimal quotient, for there are classes of LTSs that do not have a unique minimal LTS. For example, the two LTSs in \Cref{fig:non_unique_min} are both minimal in size, so a canonical minimal quotient would need to be able to distinguish between minimal LTSs.

\begin{figure}[ht]
    \centering
    \begin{tikzpicture}[node distance = 1.5cm]   
        \tikzset{every path/.style={thick}}     
        \node[node] (rpp)[below=of r] {$r$};
        \node[node] (qpp)[left=of rpp] {$q$};
        \node[node] (ppp)[left=of qpp] {$p$};
        \coordinate[left=0.5cm of ppp] (initLpp);
        \draw[->] (initLpp) to (ppp);

        \draw[->, bend left] (ppp) to["$a$"] (qpp);
        \draw[->, bend left] (qpp) to["$a$"] (rpp);
        \draw[->, bend right] (ppp) to["$\tau$" below] (qpp);
        \draw[->, bend right] (qpp) to["$\tau$" below] (rpp);
        \draw[->, loop right] (rpp) to["$b$" right] (rpp);
        
        \node[node] (pppp)[right=2.5cm of rpp] {$p$};
        \node[node] (qppp)[right=of pppp] {$q$};
        \node[node] (rppp)[right=of qppp] {$r$};
        \coordinate[left=0.5cm of pppp] (initLppp);
        \draw[->] (initLppp) to (pppp);
        
        \draw[->] (pppp) to["$a$"] (qppp);
        \draw[->, bend left=45] (qppp) to["$a$"] (rppp);
        \draw[->, bend right=45] (pppp) to["$\tau$" below] (rppp);
        \draw[->] (qppp) to["$\tau$"] (rppp);
        \draw[->, loop right] (rppp) to["$b$" right] (rppp);
    \end{tikzpicture}
    \caption{The two LTSs are equivalent and minimal. In particular, we cannot map between them using the transformations listed so far.}
    \label{fig:non_unique_min}
\end{figure}
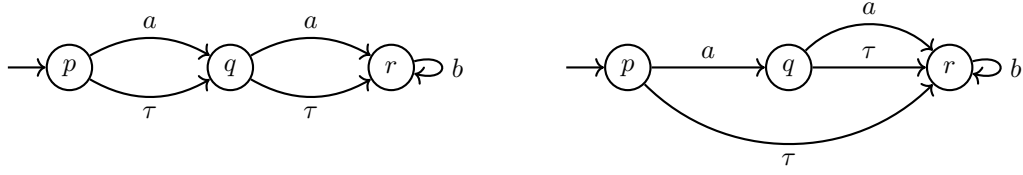

\section{Desaturation and Canonicity}\label{sec:can}


The previous sections showed that transformations resembling those found in the literature are insufficient to obtain a canonical quotient.
It turns out that with the introduction of one more transformation, $\tau$-desaturation, we \emph{are} able to describe a quotient that is canonical. This transformation is based on the axioms $\tau.x=x$ and $\tau.(\tau.x+y)=\tau.x+y$ for weak simulation equivalence and coupled simulation, respectively~\cite{aceto_axiomatizing_2014,cs-axiom}. They describe the removal of a $\tau$-transition followed by the saturation of the source state with the outgoing transitions of the target state, see \Cref{fig:path_reduction}. For coupled similarity, this requires that the target state also has an outgoing $\tau$-transition.

\begin{figure}[ht]
    \centering
    \begin{tikzpicture}[node distance = 1.5cm]
        \tikzset{every path/.style={thick}}
        \node[node] (p) {$p$};
        \node[node] (q)[right=of p] {$q$};
        \node[node] (s)[below=of p] {$s$};
        \node[node] (r)[right=of s] {$r$};
        \coordinate[left=0.5cm of p] (initL);

        \draw[->, bend left] (p) to["$\tau$"] (q);
        \draw[->, bend right] (p) to["$a$" below] (q);
        \draw[->, bend left=15] (q) to["$\tau$"] (s);
        \draw[->, bend left=15] (q) to["$a$"] node[midway](x){} (r);

        \draw[->] (initL) to (p);
        
        \node[node] (pp)[right=2.5cm of q] {$p$};
        \node[node] (qp)[right=of pp] {$q$};
        \node[node] (sp)[below=of pp] {$s$};
        \node[node] (rp)[right=of sp] {$r$};
        \coordinate[left=0.5cm of pp] (initLp);

        \draw[->] (pp) to["$a$"] (qp);
        \draw[->, bend right=15] (pp) to["$\tau$" left] node[midway](y){} (sp);
        \draw[->, bend left=15] (pp) to["$a$" below right=0.1cm and 0.3cm] (rp);
        \draw[->, bend right=15] (qp) to["$\tau$" below left=0.1cm and 0.3cm] (sp);
        \draw[->, bend left=15] (qp) to["$a$"] (rp);

        \draw[->] (initLp) to (pp);

        \draw[|->, shorten <=0.75cm, shorten >=0.75cm] (x) to (y);
    \end{tikzpicture}
    \caption{$\tau$-Desaturation maps the left LTS to the right LTS by removing the $p\xrightarrow{\tau}q$ transition and adding $p\xrightarrow{\alpha}q'$ for each $q\xrightarrow{\alpha}q'$. For coupled similarity, we require that $q$ has an outgoing $\tau$-transition.}
    \label{fig:path_reduction}
\end{figure}
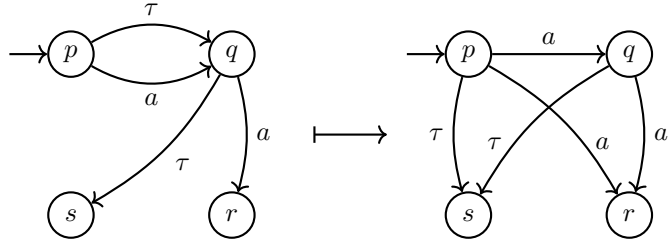

\begin{restatable}[$\tau$-Desaturation]{proposition}{ptau}\label{prop:tau}
    Let $L=\langle S, \act_\tau, \rightarrow, \iota\rangle$ be an LTS without inert transitions. Let $p,q\in S$ be such that $p\xrightarrow{\tau}q$. Assume $x=CS$ implies $q\xrightarrow{\tau}r$ for some $r\in S$. Define: 
    \[
        \bar{L}\coloneqq (L\cup\setof{(p,\alpha,q')}{q\xrightarrow{\alpha}q'})\setminus\set{(p,\tau,q)}\quad.
    \]
    Then $\bar{L}\equiv_x L$. Moreover, for all $s\in S$, it holds that $s$ in $L$ is equivalent to $s$ in $\bar{L}$.	
\end{restatable}

Since $\tau$-desaturation may introduce new $\tau$-transitions, it is not immediately clear whether or not the repeated application of $\tau$-desaturation terminates. The following proposition states that it does. As a consequence, we may obtain an LTS with no $\tau$-transitions in the case of weak simulation equivalence, or no \emph{consecutive} $\tau$-transitions in the case of coupled simulation. We call such an LTS \emph{desaturated}.

\begin{restatable}{proposition}{termination}\label{prop:termination}
    Let $L=\langle S, \act_\tau, \rightarrow, \iota\rangle$ be an LTS without inert transitions. The repeated application of $\tau$-desaturation terminates.	
\end{restatable}

Reconsider the LTSs in \Cref{fig:not_can_min}. Applying $\tau$-desaturation on the left LTS to remove the $p\xrightarrow{\tau}q$ transition results in the state $q$ being unreachable. This unreachable state can subsequently be removed, and as such we obtain the LTS on the right of \Cref{fig:not_can_min}. Likewise for the LTSs in \Cref{fig:non_unique_min}, we can apply $\tau$-desaturation on the left LTS to remove the $p\xrightarrow{\tau}q$ transition. One of the added transitions is a $p\xrightarrow{a}r$ transition, but this transition is covered (since $p\xrightarrow{\tau}q\xrightarrow{a}r$) and thus it can be removed, resulting in the LTS on the right of \Cref{fig:non_unique_min}.

We now have the necessary transformations to obtain a canonical quotient for $\equiv_x$. We say an LTS is \emph{reduced} if:
\begin{itemize}
    \item it is desaturated;
    \item it has no covered transitions;
    \item it has no unreachable states;
    \item it has no inert transitions; and
    \item no two states are equivalent.
\end{itemize}

It turns out that whenever two equivalent LTSs are reduced, they are also isomorphic, as claimed by~\Cref{thm:can}.
The proof of this theorem uses the following observation:
\begin{restatable}{lemma}{simequiv}\label{lem:sim_equiv}
    Let $L=\langle S, \act_\tau, \rightarrow, \iota\rangle$ be a desaturated LTS with no covered transitions. Let $p\xrightarrow{\alpha}q$ be some transition of $L$ and let $p'\equiv_x p$. Then there exists some state $q'\equiv_x q$ such that $p'\xRightarrow{\alpha}q'$.	
\end{restatable}
\begin{theorem}\label{thm:can}
    Let $L=\langle S, \act_\tau, \rightarrow, \iota\rangle$ and $\bar{L}=\langle \bar{S}, \act_\tau, \bar{\rightarrow}, \bar{\iota}\rangle$ be two reduced LTSs. Then $L\equiv_x \bar{L}$ iff $L\sim \bar{L}$.
\end{theorem}
\begin{proof}
    It is clear $L\sim \bar{L}$ implies $L\equiv_x \bar{L}$. To show that $L\equiv_x \bar{L}$ implies $L\sim \bar{L}$, we use Lemma 21 from~\cite{eloranta_essential_1997} which states that if:
    \begin{enumerate}
        \item $L\bisim_w\bar{L}$, and
        \item $L$ and $\bar{L}$ are state minimal w.r.t. $\bisim_w$, and
        \item $L$ and $\bar{L}$ have no duplicate transitions,
    \end{enumerate}
    then $L\sim\bar{L}$. For (1), it immediately follows by \Cref{lem:sim_equiv} that 
	\[
        \RR\coloneqq\setof{(p,\bar{q})\in S\times\bar{S}}{p\equiv_x\bar{q}}\cup\setof{(\bar{p},q)\in \bar{S}\times S}{\bar{p}\equiv_x q}
    \]
    is a weak bisimulation. For (2), state minimality for weak bisimilarity requires that no two states are weakly bisimilar, and all states are reachable~\cite{eloranta_essential_1997}. Since $\bisim_w{\subseteq}\equiv_x$~\cite{van_glabbeek_linear_1993}, and no two states are equivalent w.r.t. $\equiv_x$, it holds that $L$ and $\bar{L}$ are state minimal w.r.t. $\bisim_w$. For (3), note that $L$ and $\bar{L}$ have no covered transitions, so it suffices to show that all duplicate transitions are covered. Consider some $p\xrightarrow{\alpha}q$ that duplicates $p'\xrightarrow{\alpha}q'$. Then by \Cref{lem:tau_succ}, it holds that $q\sqsubseteq_x q'$, and therefore $p\xrightarrow{\alpha}q$ is covered by $p'\xrightarrow{\alpha}q'$. Hence, $p\xrightarrow{\alpha}q$ is covered.
\end{proof}

\Cref{alg:non_min_quotient} computes a reduced LTS equivalent to the input LTS. The order of operations is important: it is essential to first compute the $\forall$-quotient in order to eliminate inert transitions and ensure the correctness of the remaining transformations; we need to desaturate before removing covered transitions, since covered transitions may be introduced by desaturation; and we need to remove covered transitions before removing unreachable states, since states may become unreachable by removing covered transitions.

\begin{algorithm}[H]
\caption{A quotienting function that yields a canonical quotient.}
\label{alg:non_min_quotient}
\begin{algorithmic}[1]
\Procedure{CanQuotient}{$L$}
    \Let{$L_1$}{$L_{/\equiv}$}\label{line:quotient}
    \Let{$L_2$}{Desaturate $L_1$}
    \Let{$L_3$}{Remove covered transitions from $L_2$}
    \Let{$L_4$}{Remove unreachable states from $L_3$}\label{line:unreachable}
    \Return{$L_4$}
\EndProcedure
\end{algorithmic}
\end{algorithm}

This algorithm can be implemented to run in polynomial time since $\sqsubseteq_x$ can be computed in polynomial time~\cite{bisping_computing_2019}, desaturation needs to look at no more than two consecutive transitions, and the removal of covered transitions needs to look at no more than two consecutive transitions due to the preceding desaturation step.

\section{Saturation and Minimality}\label{sec:min}

The procedure described in the previous section does not necessarily yield a \emph{minimal} quotient. This is because $\tau$-desaturation may cause a strict increase in the number of transitions, as shown in \Cref{fig:desaturate_non_min}. Yet, $\tau$-desaturation is needed because it reduces the number of incoming transitions to a state, and if this number is reduced to zero, the size of the LTS can be reduced by removing the state.

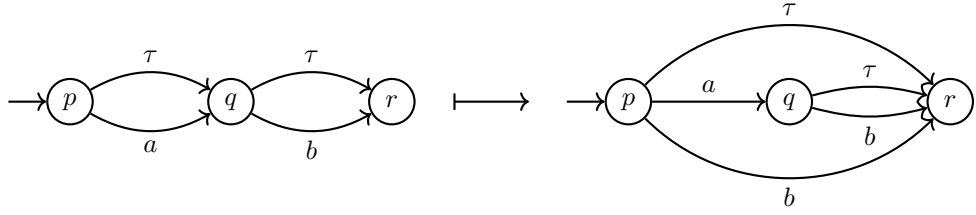
\begin{figure}[ht]
    \centering
    \begin{tikzpicture}[node distance = 1.5cm]
        \tikzset{every path/.style={thick}}
        \node[node] (p) {$p$};
        \node[node] (q)[right=of p] {$q$};
        \node[node] (r)[right=of q] {$r$};
        
        \node[node] (pp) [right=2.5cm of r] {$p$};
        \node[node] (qq) [right=of pp] {$q$};
        \node[node] (rr) [right=of qq] {$r$};
        
        \draw[->, bend right] (p) to["$a$" below] (q);
        \draw[->, bend right] (q) to["$b$" below] (r);
        \draw[->, bend left] (p) to["$\tau$"] (q);
        \draw[->, bend left] (q) to["$\tau$"] (r);
        
        \draw[|->, shorten <=0.5cm, shorten >=1cm] (r) to (pp);
        
        \draw[->] (pp) to["$a$"] (qq);        
        \draw[->, bend left=15] (qq) to["$\tau$"] (rr);
        \draw[->, bend right=15] (qq) to["$b$" below] (rr);
        
        \draw[->, bend left=45] (pp) to["$\tau$"] (rr);
        \draw[->, bend right=45] (pp) to["$b$" below] (rr);
        
        \coordinate[left=0.5cm of p] (initL);
        \coordinate[left=0.5cm of pp] (initLc);
        \draw[->] (initL) to (p);
        \draw[->] (initLc) to (pp);
    \end{tikzpicture}
    \caption{$\tau$-Desaturation maps the left LTS to the right LTS.}
    \label{fig:desaturate_non_min}
\end{figure}

In this section, we provide a procedure for obtaining a minimal quotient for $\equiv_x$, the key to which is the inverse operation of $\tau$-desaturation. Consider \Cref{fig:desaturate_non_min}: to go from the right LTS to the left LTS, observe that one may reintroduce the $p\xrightarrow{\tau}q$ transition, and then remove the covered $p\xrightarrow{\tau}r$ and $p\xrightarrow{b}r$ transitions. The introduction of the $p\xrightarrow{\tau}q$ is known as $\tau$\emph{-saturation}, and this transformation is sound since $q\sqsubseteq_x p$.

\begin{restatable}[$\tau$-Saturation]{proposition}{saturate}\label{prop:saturate}
    Let $L=\langle S, \act_\tau, \rightarrow, \iota\rangle$ be an LTS. Let $p,q\in S$ be such that $q\sqsubseteq_x p$, and define $\bar{L}\coloneqq L\cup\set{(p,\tau,q)}$. Then $L\equiv_x\bar{L}$. Moreover, for all $s\in S$, it holds that $s$ in $L$ is equivalent to $s$ in $\bar{L}$.	
\end{restatable}

A reduction is then given by performing $\tau$-saturation whenever doing so would result in covered transitions that can be removed. We say a transition $p\xrightarrow{\alpha}q$ of an LTS, $L$, is \emph{coverable} if it can be made covered by a $\tau$-saturation step, i.e. if it is an element of:
\[
    \Removable(p)\coloneqq\setof{(p,\alpha,q)}{\exists p'\sqsubset_x p: p\xrightarrow{\alpha}q\text{ is covered in }L\cup\set{(p,\tau,p')}}\quad.
\]
\vspace{-4ex}
\begin{remark}
In the definition above, we require $p'\sqsubset_x p$ rather than $p'\sqsubseteq_x p$ since $(p,\tau,p')$ cannot be inert.
\end{remark}

The existential quantifier in this definition strongly suggests that there may not be a unique, deterministic way to remove coverable transitions, for there may be several $\tau$-saturation steps that can be used to cover a transition. This raises the following question: in case we are given such a choice, is one more optimal than the other?

\begin{example}\label{eg:cover1}
    In this example, we show that this choice matters. Consider the LTS in \Cref{fig:saturate_choice}. Here, it holds that each $s_i\sqsubset_x u$, and consequently we may $\tau$-saturate by adding $u\xrightarrow{\tau}s_i$. Note that if we add a $u\xrightarrow{\tau}s_i$ transition, then each $u\xrightarrow{\alpha}\bot$ transition for which $\alpha\in S_i$ becomes covered, and therefore may be removed. Since each $u\xrightarrow{\alpha}\bot$ transition is such that $\alpha$ is in some $S_i$, we may remove every such transition using $\tau$-saturation.

	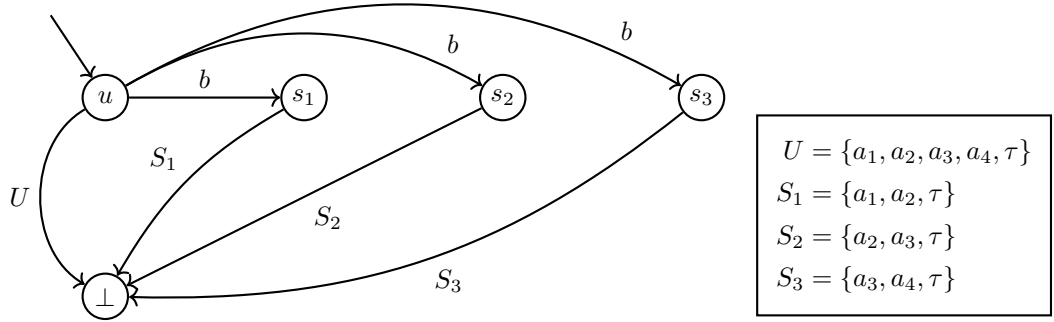
\begin{figure}[ht]
		\centering
		\begin{tikzpicture}[node distance = 2cm]
			\tikzset{every path/.style={thick}}
			\node[node] (U) {$u$};
			\node[node] (S1)[right=of U] {$s_1$};
			\node[node] (S2)[right=of S1] {$s_2$};
			\node[node] (S3)[right=of S2] {$s_3$};
			\node[node] (bot)[below=of U] {$\bot$};

			\draw[->] (U) to ["$b$"] (S1);
			\draw[->, bend left] (U) to ["$b$" very near end] (S2);
			\draw[->, bend left] (U) to ["$b$" very near end] (S3);
			\draw[->, bend right=60] (U) to ["$U$" left] (bot);
			\draw[->, bend right=15] (S1) to ["$S_1$" above left] (bot);
			\draw[->] (S2) to ["$S_2$"] (bot);
			\draw[->, bend left=20] (S3) to ["$S_3$"] (bot);

			\coordinate[above left=0.866cm and 0.5cm of U] (init);
			\draw[->] (init) to (U);

			\node[rectangle, draw, inner sep=0.25cm] (legend)[below right=0cm and 0.5cm of S3] {$\begin{aligned}
				U&=\set{a_1,a_2,a_3,a_4,\tau}\\
				S_1&=\set{a_1,a_2,\tau}\\
				S_2&=\set{a_2,a_3,\tau}\\
				S_3&=\set{a_3,a_4,\tau}
			\end{aligned}$};
		\end{tikzpicture}
		\caption{Example showing that $\tau$-saturation must be performed carefully.}
		\label{fig:saturate_choice}
	\end{figure}

    To demonstrate that the order of $\tau$-saturation is important, suppose we add a $u\xrightarrow{\tau}s_2$ transition. Then we may remove the $u\xrightarrow{\alpha}\bot$ transition for each $\alpha\in S_2$. So one transition was added, three were removed, and we are left with a $u\xrightarrow{a_1}\bot$ transition, and a $u\xrightarrow{a_4}\bot$ transition. The net effect is the removal of two transitions. At this point, $\tau$-saturation will not reduce the overall number of transitions: we could add a $u\xrightarrow{\tau}s_1$ transition, but only one transition would be removable, and likewise for $s_3$. 
    
    If instead of adding a $u\xrightarrow{\tau}s_2$ transition, we add a $u\xrightarrow{\tau}s_1$ and $u\xrightarrow{\tau}s_3$, then each $u\xrightarrow{\alpha}\bot$ transition becomes covered, so we may remove five transitions at the cost of adding two. This reduction performs better than when we added the $u\xrightarrow{\tau}s_2$ transition, so we conclude that the choice of $\tau$-saturation is important.
\end{example}

Whereas $\Removable(p)$ describes the set of transitions that can be covered using \emph{some} $\tau$-saturation step, we define $\Removed((p,\tau,p'))$ as the set of transitions that are covered by a \emph{given} $\tau$-saturation step.
\[
    \Removed((p,\tau,p'))\coloneqq\setof{(p,\alpha,q)}{p\xrightarrow{\alpha}q\text{ is covered in }L\cup\set{(p,\tau,p')}}\quad.
\]
For a given state $p$ and set $Y\subseteq\setof{(p,\alpha,q)}{p\xrightarrow{\alpha}q}$, we say a set, $X\subseteq\setof{(p,\tau,p')}{p'\sqsubset_x p}$, is a \emph{cover} of $Y$ iff $\bigcup_{(p,\tau,p')\in X}\Removed((p,\tau,p'))\supseteq Y$, and say $X$ is a cover of $p$ iff $X$ is a cover of $\Removable(p)$. What the previous example highlights is that some covers are smaller than others, and thus result in better reductions. As suggested by this example (and the terminology we use), the problem of finding a minimal cover is NP-complete, by reduction to and from the set cover problem.

\begin{definition}[Set Cover Problem]
	An instance of the set cover problem is a pair $\langle U, \mathcal{S}\rangle$ where $\bigcup_{S\in\mathcal{S}}S=U$. We call $U$ the \emph{universe} and $\mathcal{S}$ the \emph{collection}. A cover of $U$ is some subset $\mathcal{C}$ of $\mathcal{S}$ such that $\bigcup_{S\in\mathcal{C}}S=U$. A solution to the problem is some minimal cover of $U$.
\end{definition}

By considering $\langle\Removable(p),\setof{\Removed((p,\tau,p'))}{p'\sqsubset_x p}\rangle$ as an instance of the set cover problem, it is not too hard to see that an algorithm for solving the set cover problem can be used to find a minimal cover for some state $p$. Likewise, given an instance $\langle U, \mathcal{S}\rangle$ of the set cover problem, one may construct an LTS similar to the one \Cref{fig:saturate_choice}, such that a minimal cover of the initial state corresponds to a minimal set cover of the given instance. It turns out that minimisation w.r.t. $\equiv_x$ is NP-hard for similar reasons.

\begin{restatable}{theorem}{minhard}
    Minimisation w.r.t. $\equiv_x$ is NP-hard.	
\end{restatable}

\begin{algorithm}[ht]
\caption{A quotienting function for a minimal quotient.}
\label{alg:min_quotient}
\begin{algorithmic}[1]
\Procedure{MinQuotient}{$L$}
    \Let{$L_1$}{$L_{/\equiv}$}\label{line:min_quotient}
    \Let{$L_2$}{Remove covered transitions from $L_1$}\label{line:min_covered_1}
    \Let{$L_3$}{Desaturate $L_2$}\label{line:min_desaturate}
    \Let{$L_4$}{Remove covered transitions from $L_3$}\label{line:min_covered_2}
    \Let{$L_5$}{Remove unreachable states from $L_4$}\label{line:min_unreachable}
    \ForAll{$p\in S$}
        \Let{$X_p$}{be a minimal cover of $p$}\label{line:min_saturate_1}
    \EndFor
    \Let{$L_6$}{$L_5\cup\bigcup_{p\in S}X_p$}\label{line:min_saturate_2}
    \Let{$L_7$}{Remove covered transitions from $L_6$}\label{line:min_covered_3}
    \Return{$L_7$}\label{line:min_return}
\EndProcedure
\end{algorithmic}
\end{algorithm}

\Cref{alg:min_quotient} describes a procedure for computing a minimal quotient. Lines~\ref{line:min_quotient}--\ref{line:min_unreachable} essentially compute $\Call{CanQuotient}{L}$; \Aref**{line:min_covered_1} is superfluous since $L_5$ is reduced, and hence by \Cref{thm:can}, $L_5\sim\Call{CanQuotient}{L}$ holds. We include \Aref**{line:min_covered_1} as it simplifies the minimality proof, due to \Cref{prop:covered_quotient}. Then, after computing the canonical quotient, each state is covered optimally. As a consequence, an efficient algorithm for the set cover problem allows us to efficiently minimise LTSs w.r.t. $\equiv_x$.

Henceforth, for an LTS $L$, we shall use $L_i$ for $1\leq i\leq 7$ to denote the LTS in \Cref{alg:min_quotient} indicated by $L_i$. We shall parameterise $\Removable$ and $\Removed$ by $i$ as follows:
\begin{align*}
    \Removable_i(p)&\coloneqq\setof{(p,\alpha,q)}{\exists p'\sqsubset_x p: p_i\xrightarrow{\alpha}q_i\text{ is covered in }L_i\cup\set{(p,\tau,p')}}\\
    \Removed_i((p,\tau,p'))&\coloneqq\setof{(p,\alpha,q)}{p_i\xrightarrow{\alpha}q_i\text{ is covered in }L_i\cup\set{(p,\tau,p')}}\quad.
\end{align*}
After \Aref**{line:min_quotient}, the transformations used preserve $\sqsubseteq_x$, and therefore we do not find the need to parameterise $\sqsubseteq_x$ by $i$. We make the following observations:
\begin{itemize}
	\item $L_5$ is reduced, so it is unique.
	\item If $|L_7|\leq |L_2|$, then by \Cref{prop:covered_quotient}, it holds that $|L_7|\leq|L|$.
\end{itemize}
We shall use the three LTSs in \Cref{fig:injection} to demonstrate various claims, and illustrate why $|L_7|\leq |L_2|$.

\begin{figure}[ht]
	\centering
	\hfill
	\begin{tikzpicture}[node distance = 1cm]
		\tikzset{every path/.style={thick}}
		\node[node] (center) {$r$};
		\node[node] (top) at (90:1.5){$p$};
		\node[node] (bl) at (210:1.5){$q$};
		\node[node] (br) at (330:1.5){$s$};
		\coordinate[above=0.5cm of top] (init);
		\draw[->] (init) to (top);

		\draw[->, dashed, bend right] (top) to["$\tau$" above left] (bl);
		\draw[->, bend left] (center) to["$c$"] (bl);
		\draw[->, bend right] (center) to["$c$" below left] (br);
		\draw[->, bend left] (bl) to["$a$"] (center);
		\draw[->, bend right] (br) to["$b$" above right] (center);
		\draw[->, dotted, bend left] (top) to["$b$"] (center);
	\end{tikzpicture}
	\hfill
	\begin{tikzpicture}[node distance = 1cm]
		\tikzset{every path/.style={thick}}
		\node[node] (center) {$r$};
		\node[node] (top) at (90:1.5){$p$};
		\node[node] (bl) at (210:1.5){$q$};
		\node[node] (br) at (330:1.5){$s$};
		\coordinate[above=0.5cm of top] (init);
		\draw[->] (init) to (top);

		\draw[->, bend left] (center) to["$c$"] (bl);
		\draw[->, bend right] (center) to["$c$" below left] (br);
		\draw[->, bend left] (bl) to["$a$"] (center);
		\draw[->, bend right] (br) to["$b$" above right] (center);
		\draw[->, dashed, bend right] (top) to["$a$" left] (center);
		\draw[->, dotted, bend left] (top) to["$b$"] (center);
	\end{tikzpicture}
	\hfill
	\begin{tikzpicture}[node distance = 1cm]
		\tikzset{every path/.style={thick}}
		\node[node] (center) {$r$};
		\node[node] (top) at (90:1.5){$p$};
		\node[node] (bl) at (210:1.5){$q$};
		\node[node] (br) at (330:1.5){$s$};
		\coordinate[above=0.5cm of top] (init);
		\draw[->] (init) to (top);

		\draw[->, bend left] (center) to["$c$"] (bl);
		\draw[->, bend right] (center) to["$c$" below left] (br);
		\draw[->, bend left] (bl) to["$a$"] (center);
		\draw[->, bend right] (br) to["$b$" above right] (center);
		\draw[->, dotted, bend left] (top) to["$\tau$"] (br);
		\draw[->, dashed, bend right] (top) to["$\tau$" above left] (bl);
	\end{tikzpicture}
	\hfill\mbox{}
	\caption{Three weak simulation equivalent LTSs used to demonstrate various claims. They are not coupled similar. Some transitions are dashed or dotted for later reference.}
	\label{fig:injection}
\end{figure}
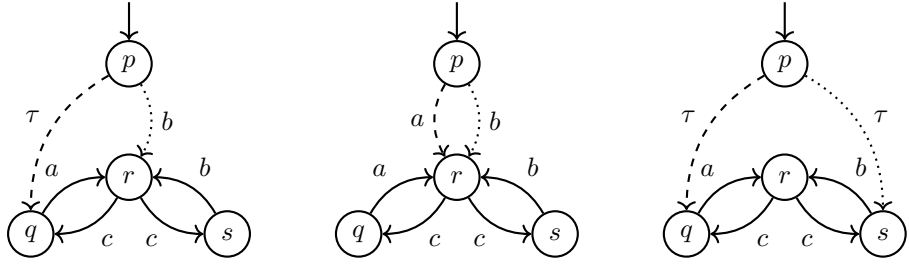

Firstly, observe the middle LTS. This LTS is reduced, so it is unique. Note that the $p\xrightarrow{a}r$ and $p\xrightarrow{b}r$ transitions are coverable. The former can be covered by adding a $p\xrightarrow{\tau}q$ transition, and the latter can be covered by adding a $p\xrightarrow{\tau}s$ transition. Hence, we indicate these pairs of transitions using dashed and dotted transitions.

Let $L$ be some LTS such that $L_2$ is given by the left LTS of \Cref{fig:injection}. Then $L_5$ is given by the middle LTS. The $p\xrightarrow{a}r$ transition was added by the $\tau$-desaturation step that removed the $p\xrightarrow{\tau}q$ transition. Reintroducing the transition allows us to remove the $p\xrightarrow{a}r$ transition again. Generally speaking, whenever there is a coverable transition of $L_5$ that is not a transition of $L_2$, then there is some $\tau$-desaturation step that can be reverted in order to remove this transition. This is given by \Cref{cor:desaturate_new}.

\begin{restatable}{lemma}{ldesaturatenew}\label{lem:desaturate_new}
    Let $L=\langle S, \act_\tau, \rightarrow, \iota\rangle$ be an LTS without inert transitions. Let $\bar{L}$ be obtained from $L$ by some number of $\tau$-desaturation steps. Suppose $p\xrightarrow{\alpha}q$ in $\bar{L}$, but not in $L$. Then there is some state $p'\neq p$ such that $p\xrightarrow{\tau}p'\xRightarrow{\alpha}q$ in $L$, but $p\nxrightarrow{\tau}p'$ in $\bar{L}$.	
\end{restatable}

\begin{restatable}{corollary}{cdesaturatenew}\label{cor:desaturate_new}
    Let $L=\langle S, \act_\tau, \rightarrow, \iota\rangle$ be an LTS such that $L_2$ and $L_7$ have the same set of states. Let $Z_p=\setof{(p,\tau,p')}{p_2\xrightarrow{\tau}p'_2\land p_3\nxrightarrow{\tau}p'_3}$. Suppose $p_5\xrightarrow{\alpha}q_5$ and $p_2\nxrightarrow{\alpha}q_2$. Then there exists some $(p,\tau,p')\in Z_p$ such that $(p,\alpha,q)\in\Removed_5((p,\tau,p'))$.	
\end{restatable}

Thus, we can use the set $Z_p$ defined in \Cref{cor:desaturate_new} to cover and remove part of $\Removable_5(p)$. When we remove transitions in this way, the LTS does not increase in size. However, not all coverable transitions get removed. In our previous example, the $p\xrightarrow{b}r$ transition does not get removed. This transition can be removed by extending $Z_p$ with the $p\xrightarrow{\tau}s$ transition, and we obtain the LTS on the right of \Cref{fig:injection}. In general, the set $Z_p$ can always be extended in this way in order to cover $p$.

\begin{restatable}{lemma}{cover}\label{lem:cover}
    Let $L=\langle S, \act_\tau, \rightarrow, \iota\rangle$ be an LTS such that $L_2$ and $L_7$ have the same set of states, and let:
    \begin{align*}
        &Y_p\text{ be a cover of }\setof{(p,\alpha,q)\in\Removable_5(p)}{p_2\xrightarrow{\alpha}q_2}\\
        &Z_p=\setof{(p,\tau,p')}{p_2\xrightarrow{\tau}p'_2\land p_3\nxrightarrow{\tau}p'_3}\quad.
    \end{align*}
    Then $Y_p\cup Z_p$ is a cover of $p$ in $L_5$.	
\end{restatable}

If the set $Y_p$ in \Cref{lem:cover} is minimal, then $Y_p$ is no larger than the subset of $\Removable_5(p)$ covered by elements of $Y_p$. Thus, removing transitions using $Y_p$ does not increase the size of the LTS. Since this was also the case for $Z_p$, we can use $Y_p\cup Z_p$ in place of $X_p$ on \Aref**{line:min_saturate_1}, and this guarantees that $|L_2|\leq|L_7|$. Then, since $X_p$ is no larger than $Y_p\cup Z_p$, and it covers all the same transitions, we also guarantee $|L_2|\leq|L_7|$ by using $X_p$.

\begin{theorem}\label{thm:min}
    Let $L$ be an LTS. Then $|L_7|\leq|L|$.
\end{theorem}
\begin{proof}
    As previously discussed, we need only show that $|L_7|\leq |L_2|$. This requires showing that if $L_2$ and $L_7$ have an equal number of states, then $L_7$ has at most as many transitions as $L_2$. Suppose $L_2$ and $L_7$ have the same number of states. Then none of the transformations after \Aref**{line:min_covered_1} modify the set of states, and hence $L_2$ and $L_7$ have the same set of states. For an arbitrary state, $p$, let $X_p$ be the set defined on \Aref**{line:min_saturate_1}, let $Y_p$ be a \emph{minimal} cover of $\hat{Y}_p=\setof{(p,\alpha,q)\in\Removable_5(p)}{p_2\xrightarrow{\alpha}q_2}$, and let $Z_p=\setof{(p,\tau,p')}{p_2\xrightarrow{\tau}p'_2\land p_3\nxrightarrow{\tau}p'_3}$. \Cref{fig:injection_2} depicts an injection, $f$, from the transitions of $L_7$ to the transitions of $L_2$, with the following form:
	\begin{align*}
        f((p,\alpha,q))&=\begin{cases}
            (p,\alpha,q)&\text{if }(p,\alpha,q)\notin X_p\\
            \rep_p \circ g_p((p,\alpha,q))&\text{if }(p,\alpha,q)\in X_p
        \end{cases}\\
		\rep_p((p,\alpha,q))&=\begin{cases}
            (p,\alpha,q)&\text{if }(p,\alpha,q)\in Z_p\setminus Y_p\\
            h_p((p,\alpha,q))&\text{if }(p,\alpha,q)\in Y_p\quad.
		\end{cases}
	\end{align*}
	
	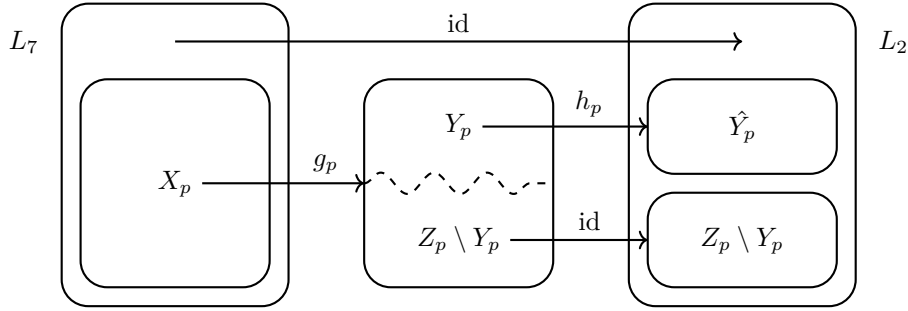
\begin{figure}[H]
		\centering
		\begin{tikzpicture}[node distance = 2cm]
			\tikzset{every path/.style={thick}}
			\draw[rounded corners=10] (0, 0) rectangle (3,4);
			\node at (-0.5,3.5) {$L_7$};
			\draw[rounded corners=10] (7.5, 0) rectangle (10.5,4);
			\node at (11,3.5) {$L_2$};

			\draw[rounded corners=10] (0.25, 0.25) rectangle (2.75,3) node[midway] (Xp) {$X_p$};
			\draw[rounded corners=10] (4, 0.25) rectangle (6.5,3) 
			node[midway, yshift=0.75cm] (Yp) {$Y_p$}
			node[midway, yshift=-0.75cm] (Zp1) {$Z_p\setminus Y_p$};
			\draw[rounded corners=10] (7.75, 0.25) rectangle (10.25,1.5) node[midway] (Zp2) {$Z_p\setminus Y_p$};
			\draw[rounded corners=10] (7.75, 1.75) rectangle (10.25,3) node[midway] (Removable) {$\hat{Y_p}$};

			\draw[snake it, dashed] (4, 1.625) -- (6.5,1.625);

			\draw[->] (Xp) to["$g_p$" xshift=16pt] (4,1.625);
			\draw[->] (Yp) to["$h_p$" xshift=9pt] (7.75,2.375);
			\draw[->] (Zp1) to["$\id$" xshift=4pt] (7.75,0.875);
			\draw[->] (1.5,3.5) to["$\id$"] (9, 3.5);
		\end{tikzpicture}
		\caption{Structure of an injection from transitions of $L_7$ to transitions of $L_2$. Each arrow represents an injection, and $\id$ is the identity function.}
		\label{fig:injection_2}
	\end{figure}

    By \Cref{lem:cover}, $Y_p\cup Z_p$ is a cover of $p_5$. Since $X_p$ is a minimal cover of $p_5$, there is an injection from $X_p$ to $Y_p\cup Z_p$. Let $g_p$ be such an injection. Now we show that there is an injection from $Y_p$ to $\hat{Y_p}$. Consider some $(p,\tau,p')\in Y_p$, and suppose for the sake of contradiction that for each $(p,\alpha,q)\in\Removed_5((p,\tau,p'))\cap\hat{Y}$, there is some $(p,\tau,p'')\in Y_p$ such that $p'\neq p''$ and $(p,\alpha,q)\in\Removed_5((p,\tau,p''))\cap\hat{Y}$. Then $Y_p\setminus\set{(p,\tau,p')}$ is a smaller cover of $\hat{Y}$, contradicting our choice of $Y_p$. Therefore, we may construct an injection from $Y_p$ to $\hat{Y}_p$ by mapping each $(p,\tau,p')\in Y_p$ to some $(p,\alpha,q)\in\hat{Y}$ that is uniquely removed by $(p,\tau,p')$. Let $h_p$ be such an injection, thereby concluding the construction of $f$.

	Now we show that $\rep_p$ is an injection. Since both $h_p$ and the identity are injections, it suffices to show that $\hat{Y_p}\cap(Z_p\setminus Y_p)=\emptyset$. By definition of $\hat{Y}_p$, each of its transitions is a transition of $L_5$. By definition of $Z_p$, each of its transitions is not a transition of $L_3$, and hence not a transition of $L_5$. Therefore, these two sets do not intersect.
	
	Now we show that $f$ is an injection. Note that $f$ maps a transition with source $p$ to another transition with source $p$. Therefore, we may fix $p$ and show that:
	\begin{enumerate}
		\item $f((p,\alpha,q))\in\hat{Y}_p$ implies $(p,\alpha,q)\in X_p$, and
		\item $f((p,\alpha,q))\in Z_p\setminus Y_p$ implies $(p,\alpha,q)\in X_p$.
	\end{enumerate}
	For (1), take $f((p,\alpha,q))\in\hat{Y}_p$ and assume for the sake of contradiction that $(p,\alpha,q)\notin X_p$. Then $f((p,\alpha,q))=(p,\alpha,q)$, and $(p,\alpha,q)$ is a transition of both $L_2$ and $L_7$. However, $\hat{Y}_p$ is a set of removable transitions in $L_5$, and therefore no transition of $\hat{Y}_p$ is a transition of $L_7$. This contradicts $p_7\xrightarrow{\alpha}q_7$. 
	
	For (2), take $f((p,\alpha,q))\in Z_p\setminus Y_p$ and assume for the sake of contradiction that $(p,\alpha,q)\notin X_p$. Then $f((p,\alpha,q))=(p,\alpha,q)$, and $(p,\alpha,q)$ is a transition of both $L_2$ and $L_7$. However, $Z_p$ is a set of transitions that do not belong to $L_3$, which means the $(p,\alpha,q)$ transition was added on \Aref**{line:min_saturate_2}, and therefore $(p,\alpha,q)\in X_p$. This contradicts our assumption that $(p,\alpha,q)\notin X_p$.

    Thus, $f$ is an injection from transitions of $L_7$ to $L_2$, which was what we wanted.
\end{proof}

\begin{restatable}{corollary}{minquotient}
    \Cref{alg:min_quotient} computes a minimal quotient.	
\end{restatable}

As discussed previously, finding minimal covers of states is NP-complete, so \Cref{alg:min_quotient} cannot be implemented efficiently unless $P=NP$. However, this procedure allows us to leverage state of the art algorithms for the set cover problem. Moreover, the instances of the set cover problem that need to be solved are typically small. Recall that finding a minimal cover of some state $p$ corresponds to finding a minimal set cover for the instance $\langle\Removable(p),\setof{\Removed((p,\tau,p'))}{p'\sqsubset_x p}\rangle$. The size of $\Removable(p)$ is bounded by the out-degree of $p$ since, in the worst case, every outgoing transition of $p$ is coverable. Thus, the size of the set that needs to be covered is small. Furthermore, by the following proposition it suffices to consider all $p'\in\max(p\sqsupset_x)$, rather than all $p'\sqsubset_x p$, both in the construction of the universe, $\Removable(p)$, and in the construction of the collection, $\setof{\Removed((p,\tau,p'))}{p'\sqsubset_x p}$. Lastly, note that $\Removed((p,\tau,p'))$ may be empty, so the size of the collection may be reduced even further.

\begin{restatable}{proposition}{covermax}\label{prop:cover_max}
	Let $L=\langle S, \act_\tau, \rightarrow, \iota\rangle$ be an LTS without inert or covered transitions, and let $p,p'\in S$ be such that $p'\sqsubset_x p$. Then $\Removed((p,\tau,p'))\subseteq\Removed((p,\tau,p''))$ holds for all $p''$ satisfying $p'\sqsubseteq_x p''\sqsubset_x p$.
\end{restatable}

The output of \Cref{alg:min_quotient} is---generally speaking---not canonical since on \Aref**{line:min_saturate_1}, there may be multiple minimal covers to choose from. To show how this step can be made deterministic, take some LTS $L$, and suppose the states of $L_5$ can be ordered $p_1,\dots,p_n$. A cover for some $p_i$ can be written as a binary array of length $n$, where the $j$th bit equals 1 iff $p_i\xrightarrow{\tau}p_j$ belongs to the cover. Then one can deterministically assign a minimal cover to $p_i$ by choosing the cover that is lexicographically minimal w.r.t. the chosen order. Since $L_5$ is canonical, we only need some algorithm that returns the `same' ordering for isomorphic LTSs. Specifically, if $L'_5\sim L_5$ and the states of $L'_5$ are ordered $q_1,\dots,q_n$, then $f:p_i\mapsto q_i$ is an isomorphism.

Algorithms for canonical forms of graphs work by computing such an ordering, and labelling vertices accordingly. It is shown in~\cite{turner_symmetry_2007} how this can be done for LTSs, assuming a universal ordering on the set of actions. In practice, this is a fair assumption since one could order actions by their unicode representations. Current algorithms for computing canonical forms have an exponential worst case complexity, so one may choose to forego this step.

\section{Conclusion}\label{sec:conc}

We investigated the problem of finding canonical and minimal quotients for coupled and weak simulation equivalence, two closely related notions of equivalence for LTSs.
While both equivalences turn out to admit canonical quotients, we showed that these are not necessarily minimal.
Furthermore, somewhat to our surprise, for both equivalences, the problem of computing a minimal quotient turns out to be NP-complete.

The quotienting function we describe for obtaining a minimal quotient essentially computes a canonical quotient and subsequently reverses some of the operations that ultimately led to an increase in size. Assuming a universal ordering on the set of actions, computing the canonical quotient as an intermediary step allows the quotienting function to compute a canonical minimal quotient, since we may take advantage of canonical labelling algorithms. 

Interestingly, our approach relied little on specific properties of the equivalences in question. For example, the proof of \Cref{thm:can} relies mainly on \Cref{lem:sim_equiv}, which shows that equivalent LTSs that are desaturated and have no covered transitions are weakly bisimilar. Afterwards, we use a previously established fact for canonicity of weakly bisimilar~LTSs. 

Our notion of $\tau$-desaturation depends on the $\tau$-law for the given equivalence. It happens to be that coupled and weak simulation equivalence have similar $\tau$-laws, which allows the results in this paper to generalise well to both equivalences. However, we cannot reasonably expect this to be the case for the remainder of the linear time-branching time spectrum. Therefore, it would be worthwhile to investigate how the $\tau$-laws for different equivalences can be used to describe canonical and minimal quotients, as they do here.




\bibliography{bib.bib}

@article{aceto_axiomatizing_2014,
  series = {Theoretical {Aspects} of {Computing} ({ICTAC} 2011)},
  title = {Axiomatizing weak simulation semantics over {BCCSP}},
  volume = {537},
  issn = {0304-3975},
  url = {https://www.sciencedirect.com/science/article/pii/S0304397513002144},
  doi = {10.1016/j.tcs.2013.03.013},
  abstract = {This paper is devoted to the study of the (in)equational theory of the largest (pre)congruences over the language BCCSP induced by variations on the classic simulation preorder and equivalence that abstract from internal steps in process behaviours. In particular, the article focuses on the (pre)congruences associated with the weak simulation, the weak complete simulation and the weak ready simulation preorders. We present results on the (non)existence of finite (ground-)complete (in)equational axiomatizations for each of these behavioural semantics. The axiomatization of those semantics using conditional equations is also discussed in some detail.},
  journal = {Theoretical Computer Science},
  author = {Aceto, Luca and de Frutos Escrig, David and Gregorio-Rodríguez, Carlos and Ingolfsdottir, Anna},
  month = {jun},
  year = {2014},
  keywords = {Complete simulation semantics, Equational axiomatizations, Equational logic, Non-finitely based algebras, Process algebra, Ready simulation semantics, Simulation semantics},
  pages = {42--71},
  timestamp = {Wed, 17 Feb 2021 00:00:00 +0100},
  biburl = {https://dblp.org/rec/journals/tcs/AcetoFGI14.bib},
  bibsource = {dblp computer science bibliography, https://dblp.org},
  _bib2doi_selected = {dblp:/rec/journals/tcs/AcetoFGI14.bib},
  _bib2doi_confirmed = {true},
}

@article{Coupledsim_Contrasim-AFP,
  author = {Benjamin Bisping and Luisa Montanari},
  title = {Coupled Similarity and Contrasimilarity, and How to Compute Them},
  journal = {Arch. Formal Proofs},
  month = {aug},
  year = {2023},
  issn = {2150-914x},
  timestamp = {Mon, 16 Oct 2023 01:00:00 +0200},
  biburl = {https://dblp.org/rec/journals/afp/BispingM23.bib},
  bibsource = {dblp computer science bibliography, https://dblp.org},
  url = {https://www.isa-afp.org/entries/Coupledsim\_Contrasim.html},
  volume = {2023},
  _bib2doi_selected = {dblp:/rec/journals/afp/BispingM23.bib},
  _bib2doi_confirmed = {true},
  _bib2doi_finished = {true},
}

@inproceedings{bisping_computing_2019,
  address = {Cham},
  title = {Computing {Coupled} {Similarity}},
  isbn = {978-3-030-17462-0},
  doi = {10.1007/978-3-030-17462-0_14},
  abstract = {Coupled similarity is a notion of equivalence for systems with internal actions. It has outstanding applications in contexts where internal choices must transparently be distributed in time or space, for example, in process calculi encodings or in action refinements. No tractable algorithms for the computation of coupled similarity have been proposed up to now. Accordingly, there has not been any tool support.},
  language = {en},
  booktitle = {Tools and {Algorithms} for the {Construction} and {Analysis} of {Systems}},
  publisher = {Springer International Publishing},
  author = {Bisping, Benjamin and Nestmann, Uwe},
  editor = {Vojnar, Tomáš and Zhang, Lijun},
  year = {2019},
  pages = {244--261},
  timestamp = {Tue, 20 Aug 2019 01:00:00 +0200},
  biburl = {https://dblp.org/rec/conf/tacas/BispingN19.bib},
  bibsource = {dblp computer science bibliography, https://dblp.org},
  _bib2doi_selected = {dblp:/rec/conf/tacas/BispingN19.bib},
  _bib2doi_confirmed = {true},
}

@inproceedings{bolognesi_fundamental_1987,
  title = {Fundamental Results for the Verification of Observational Equivalence: {A} Survey},
  shorttitle = {Fundamental {Results} for the {Verification} of {Observational} {Equivalence}},
  author = {Tommaso Bolognesi and Scott A. Smolka},
  month = {jan},
  year = {1987},
  note = {Pages: 179},
  timestamp = {Thu, 03 Jan 2002 00:00:00 +0100},
  biburl = {https://dblp.org/rec/conf/pstv/BolognesiS87.bib},
  bibsource = {dblp computer science bibliography, https://dblp.org},
  booktitle = {Protocol Specification, Testing and Verification VII, Proceedings of the {IFIP} {WG6.1} Seventh International Conference on Protocol Specification, Testing and Verification, Zurich, Switzerland, 5-8 May, 1987},
  publisher = {North-Holland},
  pages = {165--179},
  editor = {Harry Rudin and Colin H. West},
  _bib2doi_selected = {dblp:/rec/conf/pstv/BolognesiS87.bib},
  _bib2doi_confirmed = {true},
  _bib2doi_finished = {true},
}

@inproceedings{simulation,
  address = {Berlin, Heidelberg},
  title = {Simulation {Based} {Minimization}},
  isbn = {978-3-540-45101-3},
  doi = {10.1007/10721959_20},
  abstract = {This work presents a minimization algorithm. The algorithm receives a Kripke structure M and returns the smallest structure that is simulation equivalent to M. The simulation equivalence relation is weaker than bisimulation but stronger than the simulation preorder. It strongly preserves ACTL and LTL (as sub-logics of ACTL*).},
  language = {en},
  booktitle = {Automated {Deduction} - {CADE}-17},
  publisher = {Springer},
  author = {Bustan, Doron and Grumberg, Orna},
  editor = {McAllester, David},
  year = {2000},
  keywords = {Equivalence Class, Kripke Structure, Model Check, Space Complexity, Temporal Logic},
  pages = {255--270},
  timestamp = {Sun, 21 May 2017 01:00:00 +0200},
  biburl = {https://dblp.org/rec/conf/cade/BustanG00.bib},
  bibsource = {dblp computer science bibliography, https://dblp.org},
  _bib2doi_selected = {dblp:/rec/conf/cade/BustanG00.bib},
  _bib2doi_confirmed = {true},
}

@article{eloranta_minimizing_1991,
  title = {Minimizing the number of transitions with respect to observation equivalence},
  volume = {31},
  issn = {1572-9125},
  url = {https://doi.org/10.1007/BF01933173},
  doi = {10.1007/BF01933173},
  abstract = {Labeled transition systems (lts) provide an operational semantics for many specification languages. In order to abstract unrelevant details of lts's, manybehavioural equivalences have been defined; here observation equivalence is considered. We are interested in the following problem:Given a finite lts, which is the minimal observation equivalent lts corresponding to it?},
  language = {en},
  number = {4},
  journal = {BIT Numerical Mathematics},
  author = {Eloranta, Jaana},
  month = {dec},
  year = {1991},
  keywords = {F.1.1, F.3.1, Labeled transition system, Minimization, Observation equivalence, Uniqueness},
  pages = {576--590},
  timestamp = {Tue, 22 Jun 2021 01:00:00 +0200},
  biburl = {https://dblp.org/rec/journals/bit/Eloranta91.bib},
  bibsource = {dblp computer science bibliography, https://dblp.org},
  _bib2doi_selected = {dblp:/rec/journals/bit/Eloranta91.bib},
  _bib2doi_confirmed = {true},
}

@article{eloranta_essential_1997,
  title = {Essential transitions to bisimulation equivalences},
  volume = {179},
  issn = {0304-3975},
  url = {https://www.sciencedirect.com/science/article/pii/S0304397596002812},
  doi = {10.1016/S0304-3975(96)00281-2},
  abstract = {Minimization of a labelled transition system (Its) is useful e.g. while condensing the global state space of a concurrent system compositionally for verification. In this paper new minimality results for both weak and branching bisimilarities are proven. It is well known that an equivalent Its with the minimal number of states can be, in the case of bisimilarities, found by identifying all equivalent states of an Its. However, the question has been partially open whether an equivalent Its with the minimal number of states and transitions can be found. We give a proof that for every weak-image-finite Its there is a unique bisimilar Its that contains the minimal number of states and transitions. We study divergence preserving bisimilarities, since divergence — i.e. the possibility to use system resources infinitely without any output — should not be ignored when liveness properties of systems have to be checked. Our results are shown to be valid also for commonly used divergence ignoring bisimilarities, weak bisimilarity and branching bisimilarity.},
  number = {1},
  journal = {Theoretical Computer Science},
  author = {Eloranta, Jaana and Tienari, Martti and Valmari, Antti},
  month = {jun},
  year = {1997},
  pages = {397--419},
  timestamp = {Wed, 17 Feb 2021 00:00:00 +0100},
  biburl = {https://dblp.org/rec/journals/tcs/ElorantaTV97.bib},
  bibsource = {dblp computer science bibliography, https://dblp.org},
  _bib2doi_selected = {dblp:/rec/journals/tcs/ElorantaTV97.bib},
  _bib2doi_confirmed = {true},
}

@incollection{jansen_om_2020,
  title = {An {O}(m log n) algorithm for branching bisimilarity on labelled transition systems},
  isbn = {978-3-030-45236-0},
  doi = {10.1007/978-3-030-45237-7_1},
  author = {Jansen, David and Groote, Jan Friso and Keiren, Jeroen and Wijs, Anton},
  month = {apr},
  year = {2020},
  pages = {3--20},
  timestamp = {Fri, 14 May 2021 01:00:00 +0200},
  biburl = {https://dblp.org/rec/conf/tacas/JansenGKW20.bib},
  bibsource = {dblp computer science bibliography, https://dblp.org},
  _bib2doi_selected = {dblp:/rec/conf/tacas/JansenGKW20.bib},
  _bib2doi_confirmed = {true},
}

@article{kanellakis_ccs_1990,
  title = {{CCS} expressions, finite state processes, and three problems of equivalence},
  volume = {86},
  issn = {0890-5401},
  url = {https://www.sciencedirect.com/science/article/pii/089054019090025D},
  doi = {10.1016/0890-5401(90)90025-D},
  number = {1},
  urldate = {2026-04-08},
  journal = {Information and Computation},
  author = {Kanellakis, Paris C. and Smolka, Scott A.},
  month = {may},
  year = {1990},
  pages = {43--68},
  timestamp = {Fri, 12 Feb 2021 00:00:00 +0100},
  biburl = {https://dblp.org/rec/journals/iandc/KanellakisS90.bib},
  bibsource = {dblp computer science bibliography, https://dblp.org},
  _bib2doi_selected = {dblp:/rec/journals/iandc/KanellakisS90.bib},
  _bib2doi_confirmed = {true},
}

@inproceedings{cs-axiom,
  author = {Joachim Parrow and Peter Sj{\"{o}}din},
  editor = {Patrice Enjalbert and Ernst W. Mayr and Klaus W. Wagner},
  title = {The Complete Axiomatization of Cs-congruence},
  booktitle = {{STACS} 94, 11th Annual Symposium on Theoretical Aspects of Computer Science, Caen, France, February 24-26, 1994, Proceedings},
  year = {1994},
  publisher = {Springer},
  address = {Berlin, Heidelberg},
  pages = {557--568},
  isbn = {978-3-540-48332-8},
  timestamp = {Sat, 20 May 2017 01:00:00 +0200},
  biburl = {https://dblp.org/rec/conf/stacs/ParrowS94.bib},
  bibsource = {dblp computer science bibliography, https://dblp.org},
  doi = {10.1007/3-540-57785-8_171},
  volume = {775},
  url = {https://doi.org/10.1007/3-540-57785-8\_171},
  series = {Lecture Notes in Computer Science},
  _bib2doi_selected = {dblp:/rec/conf/stacs/ParrowS94.bib},
  _bib2doi_confirmed = {true},
  _bib2doi_finished = {true},
}

@article{reniers_results_2014,
  title = {Results on {Embeddings} {Between} {State}-{Based} and {Event}-{Based} {Systems}},
  volume = {57},
  issn = {0010-4620},
  url = {https://doi.org/10.1093/comjnl/bxs156},
  doi = {10.1093/comjnl/bxs156},
  abstract = {Kripke Structures (KSs) and Labelled Transition Systems (LTSs) are the two most prominent semantic models used in concurrency theory. Both models are commonly believed to be equi-expressive. One can find many ad hoc embeddings of one of these models into the other. We build upon the seminal work of De Nicola and Vaandrager that firmly established the correspondence between stuttering equivalence in KSs and divergence-sensitive branching bisimulation in LTSs. We show that their embeddings can also be used for a range of other equivalences of interest, such as strong bisimilarity, simulation equivalence and trace equivalence. Furthermore, we extend the results by De Nicola and Vaandrager by showing that there are additional translations that allow one to use minimization techniques in one semantic domain to obtain minimal representatives in the other semantic domain for these equivalences.},
  number = {1},
  urldate = {2026-04-08},
  journal = {The Computer Journal},
  author = {Reniers, Michel A. and Schoren, Rob and Willemse, Tim A.C.},
  month = {jan},
  year = {2014},
  pages = {73--92},
  timestamp = {Wed, 25 Sep 2019 01:00:00 +0200},
  biburl = {https://dblp.org/rec/journals/cj/ReniersSW14.bib},
  bibsource = {dblp computer science bibliography, https://dblp.org},
  _bib2doi_selected = {dblp:/rec/journals/cj/ReniersSW14.bib},
  _bib2doi_confirmed = {true},
}

@inproceedings{turner_symmetry_2007,
  author = {Turner, Edd and Leuschel, Michael and Spermann, Corinna and Butler, Michael},
  booktitle = {First Joint IEEE/IFIP Symposium on Theoretical Aspects of Software Engineering (TASE '07)},
  title = {Symmetry Reduced Model Checking for B},
  year = {2007},
  pages = {25-34},
  keywords = {State-space methods;Explosions;Automatic control;Control systems;Computer science;Packaging;Set theory;Automatic logic units;Electrical equipment industry;Industrial control},
  doi = {10.1109/TASE.2007.50},
  timestamp = {Fri, 24 Mar 2023 00:00:00 +0100},
  biburl = {https://dblp.org/rec/conf/tase/TurnerLSB07.bib},
  bibsource = {dblp computer science bibliography, https://dblp.org},
  _bib2doi_selected = {dblp:/rec/conf/tase/TurnerLSB07.bib},
  _bib2doi_confirmed = {true},
}

@inproceedings{van_glabbeek_linear_1993,
  address = {Berlin, Heidelberg},
  title = {The linear time — {Branching} time spectrum {II}},
  isbn = {978-3-540-47968-0},
  doi = {10.1007/3-540-57208-2_6},
  language = {en},
  booktitle = {{CONCUR}'93},
  publisher = {Springer},
  author = {van Glabbeek, Rob J.},
  editor = {Best, Eike},
  year = {1993},
  keywords = {Couple Simulation, Explicit Divergence, Global Testing, Process Algebra, Visible Action System},
  pages = {66--81},
  timestamp = {Sat, 20 May 2017 01:00:00 +0200},
  biburl = {https://dblp.org/rec/conf/concur/Glabbeek93.bib},
  bibsource = {dblp computer science bibliography, https://dblp.org},
  _bib2doi_selected = {dblp:/rec/conf/concur/Glabbeek93.bib},
  _bib2doi_confirmed = {true},
}

@inproceedings{DBLP:conf/icalp/KuceraM99,
  author       = {Anton{\'{\i}}n Kucera and
                  Richard Mayr},
  editor       = {Jir{\'{\i}} Wiedermann and
                  Peter van Emde Boas and
                  Mogens Nielsen},
  title        = {Simulation Preorder on Simple Process Algebras},
  booktitle    = {Automata, Languages and Programming, 26th International Colloquium,
                  ICALP'99, Prague, Czech Republic, July 11-15, 1999, Proceedings},
  series       = {Lecture Notes in Computer Science},
  pages        = {503--512},
  publisher    = {Springer},
  year         = {1999},
  url          = {https://doi.org/10.1007/3-540-48523-6\_47},
  doi          = {10.1007/3-540-48523-6\_47},
  bibsource    = {dblp computer science bibliography, https://dblp.org}
}

@INPROCEEDINGS{coupledsimoriginal,
  author={Ulidowski, I.},
  booktitle={[1992] Proceedings of the Seventh Annual IEEE Symposium on Logic in Computer Science}, 
  title={Equivalences on observable processes}, 
  year={1992},
  volume={},
  number={},
  pages={148-159},
  doi={10.1109/LICS.1992.185529}}

@inproceedings{park_concurrency_1981,
	address = {Berlin, Heidelberg},
	title = {Concurrency and automata on infinite sequences},
	isbn = {978-3-540-38561-5},
	doi = {10.1007/BFb0017309},
	language = {en},
	booktitle = {Theoretical {Computer} {Science}},
	publisher = {Springer},
	author = {Park, David},
	editor = {Deussen, Peter},
	year = {1981},
	pages = {167--183},
}


\appendix

\section{Proofs}\label[appendix]{sec:proofs}

\tausucc*
\begin{proof}
    We show that
    \[
    \RR\coloneqq\setof{(r,r)}{r\in S}\cup\set{(q,p)}
    \]
    is an $x$-simulation.

    \proofsubparagraph{Simulation.} Clearly a pair of the form $(r,r)\in\RR$ satisfies the simulation condition. For the pair $(q,p)\in\RR$, suppose $q\xrightarrow{\alpha}q'$. Then $p\xRightarrow{\tau}q\xrightarrow{\alpha}q'$, and thus $p\xRightarrow{\alpha}q'$. The simulation condition is satisfied since $(q',q')\in\RR$.

    \proofsubparagraph{Coupling.} We show that $x=CS$ implies $\RR$ is a coupled simulation. A pair of the form $(r,r)\in\RR$ immediately satisfies the coupling condition. For the pair $(q,p)\in\RR$, it holds that $p\xRightarrow{\tau}q$, and thus the coupling condition is satisfied since $(q,q)\in\RR$.
\end{proof}

\simact*
\begin{proof}
    Since $p\xRightarrow{\alpha}p'$, there exists states $r,r'$ such that $p\xRightarrow{\tau}r\xrightarrow{\alpha}r'\xRightarrow{\tau}p'$. By \Cref{lem:tau_succ}, it holds that $r\sqsubseteq_x p\sqsubseteq_x q$. By the simulation condition, there is some state $q'$ such that $q\xRightarrow{\alpha}q'$ and $r'\sqsubseteq_x q'$. By \Cref{lem:tau_succ} it holds that $p'\sqsubseteq_x r'\sqsubseteq_x q'$, as required.
\end{proof}

\lmax*
\begin{proof}
    Let $q\equiv_x p$ and $p'\in\max(G_\alpha(p))$. We distinguish two cases.
    \begin{enumerate}
        \item If $\alpha=\tau$, then by \Cref{lem:tau_succ}, it holds that $p'\sqsubseteq_x p$. Since $p'\in\max(G_\alpha(p))$, it holds that $p'\equiv_x p$. Therefore, since $q\xRightarrow{\tau}q$ and $q\equiv_x p$, the desired result holds.
        \item If $\alpha\neq\tau$, then by \Cref{lem:sim_act}, there is some $q'\sqsupseteq_x p'$ such that $q\xRightarrow{\alpha}q'$. Using \Cref{lem:sim_act} again, there is some $p''\sqsupseteq_x q'$ such that $p\xRightarrow{\alpha}p''$. Since $p'$ is maximal and $p'\sqsubseteq_x q'\sqsubseteq_x p''$, it holds that $p''\sqsubseteq_x p'$. Therefore, $q'\sqsubseteq_x p''\sqsubseteq_x p'$, showing that $p'\equiv_x q'$ as desired.
    \end{enumerate}
\end{proof}

\lmin*
\begin{proof}
    Let $q\sqsupseteq_x p$ and $p'\in\Min(G_\tau(p))$. By \Cref{lem:tau_succ}, it holds that $p'\sqsubseteq_x p\sqsubseteq_x q$. By the coupling condition, there exists some $q'\sqsubseteq_x p'$ such that $q\xRightarrow{\tau}q'$. By the coupling condition again, there exists some $p''\sqsubseteq_x q'$ such that $p'\xRightarrow{\tau}p''$. Then $p''\in G_\tau(p)$ and $p''\sqsubseteq_x p'$. Since $p'$ is minimal, it holds that $p'\sqsubseteq_x p''$. Therefore, $p'\sqsubseteq_x p''\sqsubseteq_x q'$, showing that $p'\equiv_x q'$ as desired.
\end{proof}

\quotient*
\begin{proof}
    We show that:
    \[
        \RR\coloneqq\setof{(p,[q]_x)}{p,q\in S\land p\sqsubseteq_x q}\cup\setof{([p]_x,q)}{p,q\in S\land p\sqsubseteq_x q}
    \]
    is an $x$-simulation.

    \proofsubparagraph{Simulation.} We show that $\RR$ is a simulation. First, take $(p,[q]_x)\in\RR$ and suppose $p\xrightarrow{\alpha}p'$. If $\alpha=\tau$, then by \Cref{lem:tau_succ}, it holds that $p'\sqsubseteq_x q$. Then $(p',[q]_x)\in\RR$ and the simulation condition is satisfied for this pair since $[q]_x\xRightarrow{\alpha}[q]_x$. If instead $\alpha\neq\tau$, then since $p\sqsubseteq q$, there is some $q'\in S$ such that $q\xRightarrow{\alpha}q'$ and $p'\sqsubseteq_x q'$. Suppose w.l.o.g. that $q'\in\Max_x(G_\alpha(q))$. Let $r\equiv_x q$ be arbitrary. By \Cref{lem:max}, there is some $r'\equiv_x q'$ such that $r\xRightarrow{\alpha}r'$. Since $r$ is arbitrary and $\alpha\neq\tau$, it holds that $[q]_x\xrightarrow{\alpha}[q']_x$. Thus, the simulation condition is satisfied for the pair $(p,[q]_x)$ since $(p',[q']_x)\in\RR$.
    
    We show that $\RR$ is a simulation. First, take $(p,[q]_x)\in\RR$ and suppose $p\xrightarrow{\alpha}p'$. Then since $p\sqsubseteq q$, there is some $q'\in S$ such that $q\xRightarrow{\alpha}q'$ and $p'\sqsubseteq_x q'$. Therefore, $(p',[q']_x)\in\RR$. Suppose w.l.o.g. that $q'\in\Max(G_\alpha(q))$. It suffices to show that $[q]_x\xRightarrow{\alpha}[q']_x$. Let $r\equiv_x q$ be arbitrary. By \Cref{lem:max}, there is some $r'\equiv_x q'$ such that $r\xRightarrow{\alpha}r'$. We distinguish two cases:
    \begin{enumerate}
        \item If $\alpha=\tau$ and $r\equiv_x r'$, then $[q]_x=[q']_x$, and thus $[q]_x\xRightarrow{\alpha}[q]_x$ holds by definition.
        \item Otherwise, since $r$ is arbitrary and $r'\equiv_x q'$ it holds that $[q]_x\xrightarrow{\alpha}[q']_x$.         
    \end{enumerate}
    Thus, the simulation condition is satisfied for this pair.

    Now, take $([p]_x,q)\in\RR$. Suppose $[p]_x\xrightarrow{\alpha}[p']_x$. Then there is some $r'\equiv_x p'$ such that $p\xRightarrow{\alpha}r'$. Since $p\sqsubseteq q$, it holds by \Cref{lem:sim_act} that $q\xRightarrow{\alpha}q'$ for some $q'$ such that $r'\sqsubseteq_x q'$. Thus, the simulation condition is satisfied for the pair $([p]_x,q)$ since $([r']_x,q')\in\RR$ and $[r']_x=[p']_x$.

    \proofsubparagraph{Coupling.} We now show that $\RR$ is a coupled simulation if $x=CS$. First, take $(p,[q]_x)\in\RR$. Since $p\sqsubseteq_x q$, it holds that $q\xRightarrow{\tau}q'$ for some $q'\sqsubseteq_x p$. Therefore, $([q']_x,p)\in\RR$. Suppose w.l.o.g. that $q'\in\Min(G_\tau(q))$. It suffices to show that $[q]_x\xRightarrow{\tau}[q']_x$. Let $r\equiv_x q$ be arbitrary. By \Cref{lem:min}, there is some $r'\equiv_x q'$ such that $r\xRightarrow{\tau}r'$. We distinguish two cases:
    \begin{enumerate}
        \item If $r\equiv_x r'$, then $[q]_x=[q']_x$, and thus $[q]_x\xRightarrow{\tau}[q]_x$ holds by definition.
        \item Otherwise, since $r$ is arbitrary and $r'\equiv_x q'$ it holds that $[q]_x\xrightarrow{\tau}[q']_x$. 
    \end{enumerate}
    Thus, the coupling condition is satisfied for this pair.

    Now, take $([p]_x,q)\in\RR$. Since $p\sqsubseteq_x q$, it holds that $q\xRightarrow{\tau}q'$ for some $q'\sqsubseteq_x p$. Then the coupling condition is satisfied for this pair since $(q',[p]_x)\in\RR$.

    \proofsubparagraph{Quotienting.} We have shown $\RR$ is an $x$-simulation. It holds that $p\equiv_x [p]_x$ for all $p\in S$ since $p\RR [p]_x$ and $[p]_x\RR p$ for all $p\in S$. Therefore, the function mapping $L$ to $L_{/\equiv}$ is a quotienting function since $[p]_x\equiv_x[q]_x$ implies $p\equiv_x q$, which implies $[p]_x=[q]_x$.
\end{proof}

\pcovered*
\begin{proof}
    The proof from right to left is clear since every transition of $L\setminus\set{(p,\alpha,q)}$ is a transition of $L$. For the proof from left to right, suppose $s\xRightarrow{\tau}p'\xrightarrow{\alpha}q'\xRightarrow{\tau}t$ in $L$. Then, there exist states $s_0,\dots,s_m$ and $t_0,\dots,t_n$ such that:
    \begin{itemize}
        \item $s_0\xrightarrow{\tau}\dots\xrightarrow{\tau}s_m$ and $(s_0,s_m)=(s,p')$; and
        \item $t_0\xrightarrow{\tau}\dots\xrightarrow{\tau}t_n$, and $(t_0,t_n)=(q',t)$.
    \end{itemize}
    We show that if any of these transitions equals the $p\xrightarrow{\alpha}q$ transition, then we reach a contradiction.
    
    Suppose for the sake of contradiction that $s_i\xrightarrow{\tau}s_{i+1}=p\xrightarrow{\alpha}q$ for some $0\leq i < m$. Then $\alpha=\tau$, and by \Cref{lem:tau_succ}, it holds that $q'\sqsubseteq p'\sqsubseteq q$. However, since $q\sqsubseteq q'$, it holds that $p'\equiv q'$. Therefore, the $p'\xrightarrow{\alpha}q'$ transition is inert, contradicting the assumption that $L$ has no inert transitions.

    Suppose for the sake of contradiction $t_i\xrightarrow{\tau}t_{i+1}=p\xrightarrow{\alpha}q$ for some $0\leq i < n$. Then $\alpha=\tau$ and $q'\xRightarrow{\tau}p$. Since, $p\xRightarrow{\tau}p'\xrightarrow{\tau}q'$ by definition of a covered transition, it holds that $p'$ and $q'$ lie on a $\tau$-cycle, and hence $p'\equiv_x q'$ by \Cref{lem:tau_succ}. Therefore, the $p'\xrightarrow{\alpha}q'$ transition is inert, contradicting the assumption that $L$ has no inert transitions.
    
    So, none of these transitions equals the $p\xrightarrow{\alpha}q$ transition, and since $(p',q')\neq(p,q)$, neither does the $p'\xrightarrow{\alpha}q'$ transition. Therefore, these transitions occur in $L\setminus\set{(p,\alpha,q)}$, and hence $s\xRightarrow{\tau}p'\xrightarrow{\alpha}q'\xRightarrow{\tau}t$ in $L\setminus\set{(p,\alpha,q)}$.
\end{proof}

\ccovered*
\begin{proof}
    Let $\bar{L}\coloneqq L\setminus\set{(p,\alpha,q)}$ and $p'\xrightarrow{\alpha}q'$ be the transition of $L$ covering $p\xrightarrow{\alpha}q$. Using $\bar{s}$ to represent the state of $\bar{L}$ corresponding to the state $s$ in $L$, consider:
    \[\RR\coloneqq\setof{(s,\bar{t})}{s,t\in S\land s\sqsubseteq_x t}\cup\setof{(\bar{s},t)}{s,t\in S\land s\sqsubseteq_x t}\quad.\]
    Since $\RR$ contains both $(\iota,\bar{\iota})$ and $(\bar{\iota},\iota)$, it suffices to show that $\RR$ is an $x$-simulation.

    \proofsubparagraph{Simulation.} Take $(s,\bar{t})\in\RR$ and suppose $s\xrightarrow{\alpha}s'$. We distinguish two cases:
    \begin{itemize}
        \item If $\alpha=\tau$, then by \Cref{lem:tau_succ}, it holds that $s'\sqsubseteq_x s$. Since $s\sqsubseteq_x t$, it holds that $(s',\bar{t})\in\RR$, and therefore the simulation condition is satisfied since $\bar{t}\xRightarrow{\tau}\bar{t}$.
        \item If $\alpha\neq\tau$, then since $s\sqsubseteq t$, it holds that $t\xRightarrow{\alpha}t'$ for some $t'\in S$ such that $s'\sqsubseteq_x t'$. Since $\alpha\neq\tau$, there exist states $r,r'\in S$ such that $t\xRightarrow{\tau}r\xrightarrow{\alpha}r'\xRightarrow{\tau}t'$. If $(r,r')\neq(p,q)$, then $\bar{t}\xRightarrow{\alpha}\bar{t}'$, and thus the simulation condition is satisfied since $(s',\bar{t}')\in\RR$. Otherwise, it holds that $t\xRightarrow{\tau}r\xRightarrow{\tau}p'$, and thus by \Cref{prop:covered}, it holds that $\bar{t}\xRightarrow{\alpha}\bar{q}'$. By \Cref{lem:tau_succ}, it holds that $t'\sqsubseteq_x r'$. Since $r'=q$ and $q\sqsubseteq_x q'$, we have $s'\sqsubseteq_x q'$. Hence, $(s',\bar{q}')\in\RR$, and the simulation condition is satisfied.
    \end{itemize}
    Thus, the simulation condition is satisfied for this pair. It is trivially satisfied for pairs of the form $(\bar{s},t)\in\RR$, since each transition of $\bar{L}$ is a transition of $L$. Hence $\RR$ is a simulation.

    \proofsubparagraph{Coupling.} We now show that $\RR$ is a coupled simulation if $x=CS$. Take $(s,\bar{t})\in\RR$. Since $s\sqsubseteq_x t$, it holds that $t\xRightarrow{\tau}t'$ for some $t'\in S$ such that $t'\sqsubseteq s$. Therefore, there are states $t_0,\dots,t_n$ such that $t_0\xrightarrow{\tau}\dots\xrightarrow{\tau}t_n$ and $(t_0,t_n)=(t,t')$. We distinguish two cases:
    \begin{itemize}
        \item If there is no $0\leq i<n$ such that $t_i\xrightarrow{\tau}t_{i+1}=p\xrightarrow{\alpha}q$, then $\bar{t}\xRightarrow{\tau}\bar{t}'$, and the coupling condition is satisfied since $(\bar{t}',s)\in\RR$.
        \item If $t_i\xrightarrow{\tau}t_{i+1}=p\xrightarrow{\alpha}q$ for some $0\leq i<n$, then $t\xRightarrow{\tau}t_i\xRightarrow{\tau}p'\xrightarrow{\tau}q'$. By \Cref{lem:tau_succ}, it holds that $t'\sqsubseteq_x q$, and therefore $s\sqsubseteq_x t'\sqsubseteq_x q\sqsubseteq_x q'$. Therefore, there $q'\xRightarrow{\tau}q''$ for some $q''\sqsubseteq s'$. By \Cref{prop:covered}, it holds that $\bar{t}\xRightarrow{\tau}\bar{q}''$. Thus, the coupling condition is satisfied since $(\bar{q}'',s)\in\RR$.
    \end{itemize}
    Thus, the coupling condition is satisfied for this pair. It is trivially satisfied for pairs of the form $(\bar{s},t)\in\RR$, since each transition of $\bar{L}$ is a transition of $L$. Hence $\RR$ is a coupled simulation.
\end{proof}

\coveredpersist*
\begin{proof}
    Let $p\xrightarrow{\alpha}q$ and $s\xrightarrow{\beta}t$ be covered by $p'\xrightarrow{\alpha}q'$ and $s'\xrightarrow{\beta}t'$, respectively. As usual, let $\bar{s}$ denote the state of $\bar{L}$ corresponding to $s$ in $L$, and likewise for other states.

    Since $s\xrightarrow{\beta}t$ is covered by $s'\xrightarrow{\beta}t'$, there are states $s_0,\dots,s_n$ such that $s_0\xrightarrow{\tau}\dots\xrightarrow{\tau}s_n$ and $(s_0,s_n)=(s,s')$. We distinguish three cases.
    \begin{enumerate}
        \item If $\bar{s}_0\xrightarrow{\tau}\dots\xrightarrow{\tau}\bar{s}_n\xrightarrow{\beta}\bar{t}'$, then $\bar{s}\xrightarrow{\beta}\bar{t}$ is covered by $\bar{s}'\xrightarrow{\beta}\bar{t}'$.
        \item If there exists $0\leq i < n$ such that $s_i\xrightarrow{\tau}s_{i+1}=p\xrightarrow{\alpha}q$, then $\alpha=\tau$. Moreover, it holds that $s\xRightarrow{\tau}p\xRightarrow{\tau}p'$. Since $s_{i+1}=q$ and $s_n=s'$, it holds that $s'\sqsubseteq_x q$, and hence $s'\sqsubseteq_x q'$. Since $\bar{s}'\xrightarrow{\beta}\bar{t}'$. By \Cref{prop:covered}, it holds that $\bar{s}\xRightarrow{\tau}\bar{p}'\xrightarrow{\tau}\bar{q}'$. We distinguish two cases.
        \begin{itemize}
            \item If also $\beta=\tau$, then by \Cref{lem:tau_succ} and because there are no inert transitions, it holds that $t'\sqsubset_x s'$. Then $t\sqsubset_x q'$ since $s'\sqsubseteq_x q'$. Since $t\sqsubset_x q'$, it holds that $(s,t)\neq(p',q')$, and therefore $\bar{s}\xrightarrow{\beta}\bar{t}$ is covered by $\bar{p}'\xrightarrow{\alpha}\bar{q}'$.
            \item If $\beta\neq\tau$, then since $s'\xrightarrow{\beta}t'$ and $s'\sqsubseteq_x q'$, it holds that $q'\xRightarrow{\beta}r'$ for some $t'\sqsubseteq r'$, and hence $t\sqsubseteq r'$. By \Cref{lem:tau_succ} and since $\beta\neq\tau$, we may assume w.l.o.g. that $q'\xRightarrow{\tau}r\xrightarrow{\beta}r'$ for some state $r$. By \Cref{prop:covered} and since $\beta\neq\alpha$, it holds that $\bar{s}\xRightarrow{\tau}\bar{r}\xrightarrow{\beta}\bar{r}'$. Moreover, $s\neq r$ as otherwise $p'\xrightarrow{\alpha}q'$ is inert, but $L$ has no inert transitions. Therefore, $\bar{s}\xrightarrow{\beta}\bar{t}$ is covered by $\bar{r}\xrightarrow{\beta}\bar{r}'$.
        \end{itemize}
        In either case, $\bar{s}\xrightarrow{\beta}\bar{t}$ is covered.
        \item If $s'\xrightarrow{\beta}t'=p\xrightarrow{\alpha}q$, then $s\xRightarrow{\tau}p\xRightarrow{\tau}p'\xrightarrow{\alpha}q'$. Since $q\sqsubseteq_x q'$ and $t'=q$, it holds that $t'\sqsubseteq_x q'$. By \Cref{prop:covered}, we have $\bar{s}\xRightarrow{\tau}\bar{p}'\xrightarrow{\alpha}\bar{q}'$. Suppose for the sake of contradiction that $s=p'$. Then since $s\xRightarrow{\tau}p\xRightarrow{\tau}p'$ and there are no inert transitions, it holds that $p=p'$, which is a contradiction since $(p',q')\neq(p,q)$. Therefore, $\bar{s}\xrightarrow{\beta}\bar{t}$ is covered by $\bar{p}'\xrightarrow{\alpha}\bar{q}'$.
    \end{enumerate}
    In all three cases, $\bar{s}\xrightarrow{\beta}\bar{t}$ is a covered transition, as required.
\end{proof}

\coveredquotient*
\begin{proof}
    Note that $|\bar{S}|\leq |S|$. If $|\bar{S}|<|S|$, then $|\bar{L}|\leq |L|$ as required. If $|\bar{S}|=|S|$, observe that $\bar{S}=S_{/\equiv}$ since only transitions are removed to obtain $\bar{L}$ from $L_{/\equiv}$. Then, since $|S|=|S_{/\equiv}|$, it holds that no two states are equivalent in $L$. Therefore, by definition of $\bar{\rightarrow}$ and $\rightarrow_{/\equiv}$, it holds that each transition $s\xrightarrow{\alpha}t$ of $\bar{\rightarrow}$ is of the form $[p]_x\xrightarrow{\alpha}[q]_x$ for some $p,q\in S$, and for such a transition it holds that $p\xRightarrow{\alpha}q$ in $L$ and $\alpha=\tau$ implies $p\not\equiv_x q$.

    Consider a transition $[p]_x\xrightarrow{\alpha}[q]_x$ in $\bar{L}$. We show that $p\xrightarrow{\alpha}q$ in $L$. As shown above, $p\xRightarrow{\alpha}q$ in $L$ and $\alpha=\tau$ implies $p\not\equiv_x q$. Therefore, there exist states $p',q'\in S$ such that $p\xRightarrow{\tau}p'\xrightarrow{\alpha}q'\xRightarrow{\tau}q$. Suppose for the sake of contradiction that $(p',q')\neq(p,q)$. We distinguish two cases.
    \begin{itemize}
        \item If $p'\neq p$, then since no two states are equivalent, it holds that $[p]_x\xrightarrow{\tau}[p']_x\xrightarrow{\alpha}[q]_x$ in $L_{/\equiv}$. Moreover, $[p]_x\neq[p']_x$. Therefore, $[p]_x\xrightarrow{\alpha}[q]_x$ is covered by $[p']_x\xrightarrow{\alpha}[q]_x$.
        \item If $q'\neq q$, then by \Cref{lem:tau_succ} it holds that $q\sqsubseteq q'$. Moreover, $[p]_x\xrightarrow{\alpha}[q']_x$ is a transition of $L_{/\equiv}$, and $[q']_x\neq[q]_x$ since no two states are equivalent. Since $q\equiv_x[q]_x$ and $q'\equiv_x[q']_x$, it holds that $[q]_x\sqsubseteq_x[q']_x$. Therefore, $[p]_x\xrightarrow{\alpha}[q]_x$ is covered by $[p]_x\xrightarrow{\alpha}[q']_x$.
    \end{itemize}
    In either case, $[p]_x\xrightarrow{\alpha}[q]_x$ is a covered transition in $L_{/\equiv}$. However, by \Cref{lem:covered_persist}, this means $[p]_x\xrightarrow{\alpha}[q]_x$ is \emph{not} a transition of $\bar{L}$, contradicting the original assumption that $[p]_x\xrightarrow{\alpha}[q]_x$ in $\bar{L}$. Therefore, $p\xrightarrow{\alpha}q$, and thus each transition of $\bar{L}$ maps injectively to a transition of $L$, and hence $|\bar{L}|\leq |L|$.
\end{proof}

\ptau*
\begin{proof}
    Let $\bar{L}\coloneqq L\setminus\set{(p,\alpha,q)}$ and $p'\xrightarrow{\alpha}q'$ be the transition of $L$ covering $p\xrightarrow{\alpha}q$. Using $\bar{s}$ to represent the state of $\bar{L}$ corresponding to the state $s$ in $L$, consider:
    \[\RR\coloneqq\setof{(s,\bar{t})}{s,t\in S\land s\sqsubseteq_x t}\cup\setof{(\bar{s},t)}{s,t\in S\land s\sqsubseteq_x t}\quad.\]
    Since $\RR$ contains both $(\iota,\bar{\iota})$ and $(\bar{\iota},\iota)$, it suffices to show that $\RR$ is an $x$-simulation.

    \proofsubparagraph{Simulation.} Take $(s,\bar{t})\in\RR$ and suppose $s\xrightarrow{\alpha}s'$. We distinguish two cases. 
    \begin{enumerate}
        \item Suppose $\alpha=\tau$, then by \Cref{lem:tau_succ}, it holds that $s'\sqsubseteq_x s$, and therefore $(s',\bar{t})\in\RR$. Since $\bar{t}\xRightarrow{\alpha}\bar{t}$, the simulation condition is satisfied for this pair.
        \item Suppose instead that $\alpha\neq\tau$. Since $s\sqsubseteq_xt$, it holds that $t\xRightarrow{\alpha}t'$ for some $t'\in S$ such that $s'\sqsubseteq t'$. Therefore, $(s',\bar{t}')\in\RR$, and it suffices to show that $\bar{t}\xRightarrow{\alpha}\bar{t}'$. Since $\alpha\neq\tau$, by \Cref{lem:tau_succ} we may assume w.l.o.g. that $t\xRightarrow{\tau}t''\xrightarrow{\alpha}t'$. Therefore, there exist states $t_0,\dots,t_n$ such that $t_0\xrightarrow{\tau}\dots\xrightarrow{\tau}t_n$ and $(t_0,t_n)=(t,t'')$. We distinguish two cases.
        \begin{itemize}
            \item If there does not exist $0\leq i < n$ such that $(t_i,t_{i+1})=(p,q)$, then $\bar{t}_0\xrightarrow{\tau}\dots\xrightarrow{\tau}\bar{t}_n$. Moreover, since $\alpha\neq\tau$, we have $\bar{t}_n\xrightarrow{\alpha}\bar{t}'$, and therefore $\bar{t}\xRightarrow{\alpha}\bar{t}'$, as required.
            \item If $(t_i,t_{i+1})=(p,q)$ for some $0\leq i < n$, then note that due to the absence of inert transitions, $t_i=t_j$ implies $i=j$, so we do not revisit $t_i$ on this path. Therefore, we either have $\bar{t}_0\xrightarrow{\tau}\dots\xrightarrow{\tau}\bar{t}_i\xrightarrow{\alpha}\bar{t}'$ if $i+1=n$, or $\bar{t}_0\xrightarrow{\tau}\dots\xrightarrow{\tau}\bar{t}_i\xrightarrow{\tau}\bar{t}_{i+2}\xrightarrow{\tau}\dots\xrightarrow{\tau}\bar{t}_n\xrightarrow{\alpha}\bar{t}'$ if $i+1<n$. In any case, we have $\bar{t}\xRightarrow{\alpha}\bar{t}'$, as required.
        \end{itemize}
    \end{enumerate}
    Thus, the simulation condition is satisfied for this pair.
    
    Now, take $(\bar{s},t)\in\RR$ and suppose $\bar{s}\xrightarrow{\alpha}\bar{s}'$. We distinguish two cases.
    \begin{enumerate}
        \item If $s\xrightarrow{\alpha}s'$, then since $s\sqsubseteq_x t$, it holds that $t\xRightarrow{\alpha}t'$ for some $t'\in S$ such that $s'\sqsubseteq_x t'$. Since $(\bar{s}',\bar{t}')\in\RR$, the simulation condition is satisfied for this pair.
        \item Otherwise, it holds that $s=p$ and $s'$ is such that $q\xrightarrow{\alpha}s'$. Therefore, $s\xrightarrow{\tau}q\xrightarrow{\alpha}s'$. By \Cref{lem:sim_act}, it holds that $t\xRightarrow{\alpha}t'$ for some $t'\in S$ such that $s'\sqsubseteq_x t'$. Since $(\bar{s}',\bar{t}')\in\RR$, the simulation condition is satisfied for this pair.
    \end{enumerate}
    Thus, the simulation condition is satisfied for this pair, and $\RR$ is a simulation.

    \proofsubparagraph{Coupling.} We now show that $\RR$ is a coupled simulation if $x=CS$. Let $r\in S$ be such that $q\xrightarrow{\tau}r$, which exists by assumption since $x=CS$. Take $(s,\bar{t})\in\RR$. Since $s\sqsubseteq_x t$, it holds that $t\xRightarrow{\tau}t'$ for some $t'\in S$ such that $t'\sqsubseteq_x s$. Therefore, there exist states $t_0,\dots,t_n$ such that $t_0\xrightarrow{\tau}\dots\xrightarrow{\tau}t_n$ and $(t_0,t_n)=(t,t')$. We distinguish two cases.
    \begin{enumerate}
        \item If there does not exist $0\leq i < n$ such that $(t_i,t_{i+1})=(p,q)$, then $\bar{t}\xRightarrow{\tau}\bar{t}'$, and therefore the coupling condition is satisfied since $(\bar{t}',s)\in\RR$.
        \item If $(t_i,t_{i+1})=(p,q)$ for some $0\leq i < n$, then note that due to the absence of inert transitions, $t_i=t_j$ implies $i=j$, so we do not revisit $t_i$ on this path. We distinguish two cases.
        \begin{itemize}
            \item If $i+1=n$, then $\bar{t}_0\xrightarrow{\tau}\dots\xrightarrow{\tau}\bar{t}_i\xrightarrow{\tau}\bar{r}$. Therefore, $\bar{t}\xRightarrow{\tau}\bar{r}$. By \Cref{lem:tau_succ}, it holds that $r\sqsubseteq_x q$. Since $q=t'$ and $t'\sqsubseteq_x s$, it holds that $(\bar{r},s)\in\RR$, and thus the coupling condition is satisfied.
            \item If $i+1<n$, then $\bar{t}_0\xrightarrow{\tau}\dots\xrightarrow{\tau}\bar{t}_i\xrightarrow{\tau}\bar{t}_{i+2}\xrightarrow{\tau}\dots\xrightarrow{\tau}\bar{t}_n$. Therefore, $\bar{t}\xRightarrow{\tau}\bar{t}'$, and the coupling condition is satisfied since $(\bar{t}',s)\in\RR$.
        \end{itemize}
    \end{enumerate}
    Thus, the coupling condition is satisfied for this pair. It is immediately satisfied for pairs of the form $(\bar{s},t)\in\RR$, and hence $\RR$ is a coupled simulation.
\end{proof}

\termination*
\begin{proof}
    To prove termination, we show that the $\tau$-adjacency matrix of $L$ decreases w.r.t. some total order $<_{\mathcal{L}}$ after each $\tau$-desaturation step, provided that the matrix is indexed in a specific order, $<_S$, induced by $\sqsubseteq_x$.
    
    To construct $<_S$, we first obtain a partial order, $\preccurlyeq_x$ on the equivalence classes of $S$ by defining $[p]_x\preccurlyeq_x[q]_x$ iff $p\sqsubseteq_x q$. Then, let $\leq_x$ be some linear extension of $\preccurlyeq_x$. Lastly, let $<_S$ be some linear ordering such that $p<_S q$ only if $[p]_x\leq_x[q]_x$ and $p\neq q$. In other words, replace equivalence classes in $\leq_x$ with the states of those equivalence classes in arbitrary order. Henceforth, we enumerate the elements of $S$ as $p_1,\dots,p_n$, where $i<j$ iff $p_i<_Sp_j$.

    We show this order satisfies the property that $p_i\xrightarrow{\tau}p_j$ implies $j< i$. By \Cref{lem:tau_succ}, it holds that $p_j\sqsubseteq_x p_i$, and therefore $[p_j]_x\preccurlyeq_x[p_i]_x$. Since $p_i\xrightarrow{\tau}p_j$ is not inert, it holds that $[p_i]_x\neq[p_j]_x$, and therefore $j< i$. Thus, $p_i\xrightarrow{\tau}p_j$ implies $j< i$. Since $\tau$-desaturation preserves $\sqsubseteq_x$ and does not introduce inert transitions, the same property holds after a $\tau$-desaturation step.
    
    Define for each $p_i\in S$ the $n$-vector $v(p_i)$ where the $j$th position, $v_j(p_i)$, equals 1 if $p_i\xrightarrow{\tau}p_j$, and equals 0 otherwise. So $v(p_i)$ is some row of the $\tau$-adjacency matrix. Then, for two $n$-vectors, $x,y$, we write $x<_{\mathcal{L}} y$ iff there exists some $1\leq i\leq n$ such that $x_i < y_i$ and $x_j=y_j$ for all $i< j \leq n$. 

    We show that after each $\tau$-desaturation step, some $v(p_i)$ decreases w.r.t. $<_{\mathcal{L}}$, whereas all others remain the same. Let $\bar{L}$ be the LTS obtained from $L$ after a single $\tau$-desaturation step, say the removal of $p_i\xrightarrow{\tau}p_j$. That is, $\bar{L}= (L\cup\setof{(p_i,\alpha,p_k)}{(p_j,\alpha,p_k)\in{\rightarrow}})\setminus\set{(p_i,\tau,p_j)}$. We write $v(\bar{p}_k)$ to denote $v(p_k)$ in $\bar{L}$. For each $k\neq i$, it holds that $v(\bar{p}_k)=v(p_k)$, since the outgoing transitions of $p_k$ remain unchanged. Next, we show that $v(\bar{p}_i)<_{\mathcal{L}}v(p_i)$ by showing that:
	\begin{enumerate}
		\item $v_j(\bar{p}_i)<v_j(p_i)$, and
		\item $v_k(\bar{p}_i)=v_k(p_i)$ for all $j < k \leq n$.
	\end{enumerate}
	For (1), it holds by definition of $L$ and $\bar{L}$ that $v_j(\bar{p}_i)=0$ whereas $v_j(p_i)=1$. For (2), let $j<k\leq n$. We distinguish two cases.
	\begin{itemize}
		\item If $p_j\xrightarrow{\tau}p_k$ in $L$, then as shown above, it holds that $k<j$, which contradicts $j<k\leq n$.
		\item Otherwise, it holds that $v_k(\bar{p}_i)=v_k(p_i)$ by definition of $L$ and $\bar{L}$.
	\end{itemize}
	Therefore, $v(\bar{p}_i)<_{\mathcal{L}} v(p_i)$.

    Thus, the repeated application of $\tau$-desaturation terminates since on each iteration, some $v(p_i)$ decreases with respect to $<_{\mathcal{L}}$, all others remain the same, and each vector is bounded below by the zero vector.
\end{proof}

\simequiv*
\begin{proof}
    We distinguish two cases. Firstly, consider $\alpha\neq\tau$. We show that $q\in\Max(G_{\alpha}(p))$, so that the result holds by \Cref{lem:max}. Suppose $p\xRightarrow{\alpha}r$ for some $r\in S$ such that $q\sqsubseteq_x r$. We need to show that $r\sqsubseteq_x q$. Since $\alpha\neq\tau$, there are states $p',r'\in S$ such that $p\xRightarrow{\tau}p'\xrightarrow{\alpha}r'\xRightarrow{\tau}r$. By \Cref{lem:tau_succ}, it holds that $r\sqsubseteq_x r'$, and hence $q\sqsubseteq_x r'$. Since $L$ has no covered transitions, it holds that $(p,q)=(p',r')$, and therefore $r\sqsubseteq_x q$, as required.

    Now, consider $\alpha=\tau$. If $x=S$, then for $L$ to be desaturated, it must have no $\tau$-transitions, and thus the result holds vacuously. Otherwise, $x=CS$, and so we show that $q\in\Min(G_{\tau}(p))$, so that the result holds by \Cref{lem:min}. Suppose $p\xRightarrow{\tau}r$ for some $r\in S$ such that $r\sqsubseteq_x q$. We need to show that $q\sqsubseteq_x r$. By the coupling condition, there is some state $q'\in S$ such that $q\xRightarrow{\tau}q'$ and $q'\sqsubseteq_x r$. Since $L$ is desaturated, there are no consecutive $\tau$-transitions, and thus $q'=q$. Hence $q\sqsubseteq_x r$ as required.
\end{proof}

\saturate*
\begin{proof}
    Let $\bar{L}\coloneqq L\cup\set{(p,\tau,q)}$. Using $\bar{s}$ to represent the state of $\bar{L}$ corresponding to the state $s$ in $L$, consider:
    \[\RR\coloneqq\setof{(s,\bar{t})}{s,t\in S\land s\sqsubseteq_x t}\cup\setof{(\bar{s},t)}{s,t\in S\land s\sqsubseteq_x t}\quad.\]
    Since $\RR$ contains both $(\iota,\bar{\iota})$ and $(\bar{\iota},\iota)$, it suffices to show that $\RR$ is an $x$-simulation.

    \proofsubparagraph{Simulation.} First, take $(\bar{s},t)\in\RR$ and suppose $\bar{s}\xrightarrow{\alpha}\bar{s}'$. We distinguish two cases.
    \begin{itemize}
        \item If $s\xrightarrow{\alpha}s'$, then since $s\sqsubseteq t$, it holds that $t\xRightarrow{\alpha}t'$ for some $t'\in S$ such that $s'\sqsubseteq_x t'$. Hence, $(\bar{s}',t')\in\RR$ and the simulation condition is satisfied.
        \item Otherwise, it holds that $\bar{s}\xrightarrow{\alpha}\bar{s}'=\bar{p}\xrightarrow{\tau}\bar{q}$. Since $q\sqsubseteq_x p$ and $p=s$, it holds that $q\sqsubseteq_x t$, and hence $(\bar{q},t)\in\RR$. Therefore, the simulation condition is satisfied since $\bar{t}\xRightarrow{\alpha}\bar{t}$.
    \end{itemize}
    Thus, the simulation condition is satisfied for this pair. It is trivially satisfied for pairs of the form $(s,\bar{t})\in\RR$ since each transition of $L$ is a transition of $\bar{L}$, and therefore $\RR$ is a simulation.

    \proofsubparagraph{Coupling.} It remains to show that $x=CS$ implies $\RR$ is a coupled simulation. For a pair $(\bar{s},t)\in\RR$, it follows that since $s\sqsubseteq_x t$, there is some $t'\sqsubseteq_x s$ such that $t\xRightarrow{\tau}t'$. Therefore, the coupling condition is satisfied since $(t',\bar{s})\in\RR$. For a pair $(s,\bar{t})\in\RR$, the coupling condition is satisfied since each transition of $L$ is a transition of $\bar{L}$, and therefore $\RR$ is a coupled simulation.
\end{proof}

\minhard*
\begin{proof}
    We show that the set cover problem is reducible to minimisation w.r.t. $\equiv_x$. For our instance of the problem, let $U=\set{a_1,\dots,a_n}$ and $\mathcal{S}=\set{S_1,\dots,S_m}$ such that $\bigcup_{S_i\in\mathcal{S}}S_i=U$. We may assume w.l.o.g. that $S_i\subseteq S_j$ implies $S_i=S_j$, because from some minimal cover $\mathcal{C}$ of $U$, we may construct another minimal cover consisting of supersets of the sets in $\mathcal{C}$.

    We construct an LTS $L=\langle S, \act_\tau, \rightarrow, u\rangle$. Here, $\act_\tau=\set{a_1,\dots,a_n}\cup\set{\bar{a}_1,\dots,\bar{a}_n}\cup\set{b,\tau}$. We use $\bar{\bar{a}}_i$ to denote $a_i$. Define $\hat{U},\hat{S}_i$ for each $1\leq i\leq m$ as the least sets satisfying:
    \begin{itemize}
        \item $\hat{U}\supseteq U\cup\setof{\bar{a}_j}{a_j\in U}$;
        \item $\hat{S}_i\supseteq S_i\cup\setof{\bar{a}_j}{a_j\in S_i}$; and
        \item $x=CS$ implies $\tau\in\hat{U}$ and $\tau\in\hat{S}_i$.  
    \end{itemize}
    Then, define:
    \begin{itemize}
        \item $S=\set{u,\bot}\cup\set{s_1,\dots,s_m}$; and
        \item $\rightarrow$ is the least relation satisfying:
        \begin{itemize}
            \item $u\xrightarrow{b}s_i$ for all $1\leq i \leq m$;
            \item $u\xrightarrow{\hat{U}}\bot$; and
            \item $s_i\xrightarrow{\hat{S}_i}\bot$ for all $1\leq i\leq m$.
        \end{itemize}
    \end{itemize}
    Then $L$ has the structure shown in the figure below.
    
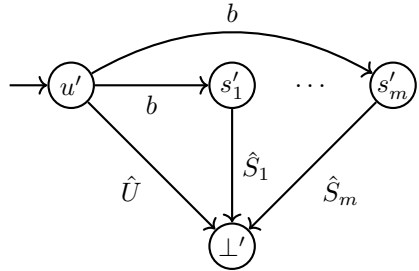
\begin{center}
    \begin{tikzpicture}[node distance = 1.5cm]
        \tikzset{every path/.style={thick}}
        \node[node] (U) {$u'$};
        \node[node] (S1)[right=of U] {$s'_1$};
        \node[node] (S2)[right=of S1] {$s'_m$};
        \node[node] (bot)[below=of S1] {$\bot'$};
        \node[] (dots) at ($(S1)!0.5!(S2)$) {$\dots$};

        \draw[->] (U) to ["$b$" below] (S1);
        \draw[->, bend left] (U) to ["$b$"] (S2);
        \draw[->] (U) to ["$\hat{U}$" below left] (bot);
        \draw[->] (S1) to ["$\hat{S}_1$"] (bot);
        \draw[->] (S2) to ["$\hat{S}_m$"] (bot);

        \coordinate[left=0.5cm of U] (init);
        \draw[->] (init) to (U);
    \end{tikzpicture}        
    \end{center}    
    Note that $L$ is reduced, and for each $1\leq i\leq m$ it holds that $\bot\sqsubset_x s_i$, and $s_i\sqsubset_x u$. Moreover, $L$ can be constructed in polynomial time since there are $m+2$ states and up to $2m+2n$ transitions. What we now do is consider some minimal LTS $L'\equiv_x L$. We show that $L'$ satisfies the following conditions:
    \begin{enumerate}
        \item For each state $p$ of $L$, there is exactly one state $p'$ in $L'$ such that $p\equiv_x p'$, and vice-versa.
        \begin{itemize}
            \item Hence, if $p$ is a state of $L$, we write $p'$ to denote the unique state of $L'$ such that $p\equiv_x p'$.
            \item Likewise, if $p'$ is a state of $L'$, we write $p$ to denote the unique state of $L$ such that $p\equiv_x p'$.
        \end{itemize}
        \item $u'\xrightarrow{b}s'_i$ for all $1\leq i\leq m$.
        \item $s'_i\xrightarrow{\hat{S}_i}\bot'$ for all $1\leq i\leq m$.
        \item There is some set $\mathcal{C}'\subseteq\setof{i}{1\leq i\leq m}$ such that $u'\xrightarrow{\tau}s'_i$ for all $i\in\mathcal{C}'$ and $\bigcup_{i\in\mathcal{C}'}S_i=U$.
        \item $L'$ has no other transitions.
    \end{enumerate}
    Lastly, we show that from $L'$, one can obtain a minimal cover of $U$ (in polynomial time) in a natural way: $\setof{S_i}{i\in\mathcal{C}'}$ is a minimal cover of $U$.

    \proofsubparagraph{Condition 1.} Note that no two states of $L'$ are equivalent, for otherwise the $\forall$-quotient of $L'$ is strictly smaller than $L'$, contradicting the choice of $L'$ as minimal. Therefore, it suffices to show that for each state $p$ of $L$, there is some state $p'$ of $L'$ such that $p\equiv_x p'$. This holds by \Cref{lem:sim_equiv}.

    \proofsubparagraph{Condition 2.} We know by \Cref{lem:sim_equiv} that $u'\xRightarrow{b}s'_i$ for all $1\leq i\leq m$. Let $1\leq i\leq m$. We know that $u'\xRightarrow{b}s'_i$, so there exist states $p',q'$ such that $u'\xRightarrow{\tau}p'\xrightarrow{b}q'\xRightarrow{\tau}s'_i$. We show that $p'=u'$ and $q'=s'_i$. To show $p'=u'$, we know that if $p'\xrightarrow{b}$, then also $p\xRightarrow{b}$. However, the only state of $L$ that can execute a weak $b$-transition is $u$, therefore, $p=u$, and thus $p'=u'$. To show $q'=s'_i$, note that by \Cref{lem:tau_succ}, it holds that $s'_i\sqsubseteq_x q'$, and hence $s_i\sqsubseteq_x q$. The only candidates for $q$ are then $s_i$ or $u$. However, if $q=u$, then this means $u\xrightarrow{b}u$, which is not the case, and thus $q=s_i$ as required.

    \proofsubparagraph{Condition 3.} We know by \Cref{lem:sim_equiv} that $s'_i\xRightarrow{\hat{S}_i}\bot'$ for all $1\leq i\leq m$. Let $1\leq i\leq m$ and $\alpha\in\hat{S}_i$. Since $s'_i\xRightarrow{\alpha}\bot'$, there exist states $p',q'$ such that $s'_i\xRightarrow{\tau}p'\xrightarrow{\alpha}q'\xRightarrow{\tau}\bot'$. To show $p'=s'_i$, note that by \Cref{lem:tau_succ}, $p'\sqsubseteq_x s'_i$. Therefore, either $p'=s'_i$, which is what we want, or $p'=\bot'$. If $p'=\bot'$, we distinguish two cases.
    \begin{itemize}
        \item If $\alpha\neq\tau$, then we have a contradiction since $\bot$ cannot execute any visible actions, and therefore neither can $\bot'$.
        \item If $\alpha=\tau$, then by \Cref{lem:tau_succ}, it holds that $q'\sqsubseteq_x p'$, but since $\bot'$ is minimal, this means that $\bot'\xrightarrow{\tau}\bot'$. Then this self-loop may be removed, contradicting the choice of $L'$ as minimal.
    \end{itemize}
    Therefore, $p'=s'_i$. To show $q'=\bot'$, suppose for the sake of contradiction that $q'\neq\bot'$, i.e. $q'=u'$ or $q'=s'_j$ for some $1\leq j\leq m$. We distinguish two cases:
    \begin{itemize}
        \item If $\alpha\neq\tau$, then note that there is some $\beta\neq\tau$ such that $q'\xrightarrow{\beta}$. But then in $L'$ there is a sequence of three visible actions, $u'\xrightarrow{b}s'_i\xrightarrow{\alpha}q'\xrightarrow{\beta}$, which is a contradiction as this is not the case in $L$, and therefore not the case in $L'$, either.
        \item If $\alpha=\tau$, then we distinguish two more cases.
        \begin{itemize}
            \item If $q'=s'_i$, then this means that $s'_i\xrightarrow{\tau}s'_i$, which means that $L'$ is not minimal, as this transition can be removed.
            \item Otherwise, this means $q'\sqsubset_x s'_i$, which is a contradiction since $q'=u'$ or $q'=s'_j$ for some $j\neq i$.
        \end{itemize}
    \end{itemize}
    Therefore $s'_i\xrightarrow{\alpha}\bot'$, but since $\alpha\in\hat{S}_i$ was arbitrary, we have $s'_i\xrightarrow{\hat{S}_i}\bot'$ as required.

    \proofsubparagraph{Condition 4.} We show that for each $\alpha\in\hat{U}$, there is some $1\leq i\leq m$ such that $\alpha\in\hat{S}_i$ and $u'\xrightarrow{\tau}s'_i$. Note that $\tau\in\hat{S}_i$ for all $1\leq i\leq m$, and since $\hat{U}\setminus\set{\tau}$ is non-empty, it suffices to consider the case for $\alpha\neq\tau$.
    
    Take $\alpha\in\hat{U}\setminus\set{\tau}$. By \Cref{lem:sim_equiv} we know there exist states $p',q'$ such that $u'\xRightarrow{\tau}p'\xrightarrow{\alpha}q'\xRightarrow{\tau}\bot'$. Likewise, there exist states $s',t'$, such that $u'\xRightarrow{\tau}s'\xrightarrow{\bar{\alpha}}t'\xRightarrow{\tau}\bot'$.
    
    First, we show that $q'=\bot'$. Suppose for the sake of contradiction that $q'\neq\bot'$, i.e. either $q'=u'$ or $q'=s'_i$ for some $1\leq i\leq m$. Then $L'$ has two consecutive $\alpha$-transitions from its initial state, $u'$. This implies $L$ has two consecutive weak $\alpha$-transitions from its initial state, but it does not, and hence we have a contradiction. Thus $q'=\bot'$, and analogously $t'=\bot'$.

    Note that $p'\neq\bot'$, for if $p'=\bot'$, then $\bot'\xrightarrow{\alpha}$, which implies that $\bot\xRightarrow{\alpha}$, which is a contradiction. We now show that $p'\neq u'$. Suppose for the sake of contradiction that $p'=u'$, so that $u'\xrightarrow{\alpha}\bot'$. We distinguish two cases.
    \begin{enumerate}
        \item Suppose $u'\xrightarrow{\tau}s'_i$ for some $s'_i$ such that $\alpha\in\hat{S}_i$ or $\bar{\alpha}\in\hat{S}_i$. Note that $\alpha\in\hat{S}_i$ iff $\bar{\alpha}\in\hat{S}_i$, so therefore $s'_i\xrightarrow{\alpha}\bot'$. Hence, the $u'\xrightarrow{\alpha}\bot'$ transition is covered, meaning $L'$ is not minimal, contradicting our choice of $L'$.
        \item Otherwise, it holds that $s'=u'$, and hence $u'\xrightarrow{\bar{\alpha}}\bot'$. Note that there is at least one $\hat{S}_i$ such that $\alpha\in\hat{S}_i$ and $\bar{\alpha}\in\hat{S}_i$, so let $\hat{S}_i$ be such. Then $u'\xrightarrow{\alpha}\bot'$ and $u'\xrightarrow{\hat{\alpha}}\bot'$ are covered in $L'\cup\set{(u',\tau,s'_i)}$. Thus, by \Cref{prop:saturate,prop:covered}, it holds that $L'\equiv_x (L'\cup\set{(u',\tau,s'_i)})\setminus\set{(u',\alpha,\bot'),(u',\bar{\alpha},\bot')}$, contradicting our choice of $L'$ as minimal.
    \end{enumerate}
    Therefore $p'\neq\bot'$, and hence $p'=s'_i$ for some $1\leq i\leq m$. Moreover, since $s'_i\xrightarrow{\alpha}$, it holds that $s_i\xRightarrow{\alpha}$, and hence $\alpha\in\hat{S}_i$ as required.

    Now, construct $\mathcal{C}'=\setof{i}{u'\xrightarrow{\tau}s'_i}$. To show $\bigcup_{i\in\mathcal{C}'}S_i=U$ is satisfied, take $\alpha\in U$. Then, $\alpha\in\hat{U}$, so there is some $s'_i$ such that $u'\xrightarrow{\tau}s'_i$ and $\alpha\in\hat{S}_i$. Let $s'_i$ be such and note that $i\in\mathcal{C}'$. Since $\alpha\in U$ and $\alpha\in\hat{S}_i$, it holds that $\alpha\in S_i$. Thus, $\bigcup_{i\in\mathcal{C}'}S_i=U$ holds.

    \proofsubparagraph{Condition 5.} Using only the necessary transitions described by conditions 1--4, it is straightforward to show that
    \[
        \RR\coloneqq\setof{(p,p')}{p\in S}\cup\setof{(p',p)}{p\in S}
    \]
    is an $x$-simulation. Therefore, no other transitions are needed in order for $L'\equiv_x L$ to be satisfied.

    \proofsubparagraph{Cover.} Let $\mathcal{C}'\subseteq\setof{i}{1\leq i\leq m}$ as described by condition 4. Condition 4 states that $\mathcal{C}=\setof{S_i}{i\in\mathcal{C}'}$ is a cover of $U$. We show that $\mathcal{C}$ is a minimal cover. Suppose for the sake of contradiction that it is not minimal, so there is some $S_i\in\mathcal{C}$ such that $\mathcal{C}\setminus\set{S_i}$ is a cover of $U$. Let $1\leq i\leq m$ be such. We show that $L'\setminus\set{(u',\tau,s'_i)}\equiv_x L'$, contradicting our choice of $L'$ as minimal. 
    
    Firstly, since $u'\xrightarrow{\tau}s'_i$ and $x=CS$ implies $s'_i\xrightarrow{\tau}\bot'$, we may apply $\tau$-desaturation to remove the $u'\xrightarrow{\tau}s'_i$ transition and introduce a $u'\xrightarrow{\alpha}p'$ transition for each $s'_i\xrightarrow{\alpha}p'$, resulting in an LTS $L''$. We show that each of these introduced transitions is covered.

    Consider some newly introduced $u'\xrightarrow{\alpha}p'$ transition. This was added because $s'_i\xrightarrow{\alpha}p'$, meaning $p'=\bot'$ and $\alpha\in\hat{S}_i$. Since $\mathcal{C}\setminus\set{S_i}$ is a cover of $U$, it holds that for each $a_j\in S_i$, there is some $S_k\in\mathcal{C}\setminus\set{S_i}$ such that $a_j\in S_k$. We distinguish three cases.
    \begin{enumerate}
        \item If $\alpha=a_j$ for some $1\leq j \leq n$, then let $S_k\in\mathcal{C}\setminus\set{S_i}$ such that $a_j\in S_k$. Then it holds that $a_j\in\hat{S}_k$, and thus $u'\xrightarrow{\tau}s'_k\xrightarrow{a_j}$ in $L''$, meaning the $u'\xrightarrow{\alpha}\bot'$ transition is covered.
        \item If $\alpha=\bar{a}_j$ for some $1\leq j \leq n$, then $a_j\in S_i$, so there is some $S_k\in\mathcal{C}\setminus\set{S_i}$ such that $a_j\in S_k$. Then it holds that $\bar{a}_j\in\hat{S}_k$. Thus, $u'\xrightarrow{\tau}s'_k\xrightarrow{\bar{a}_j}$ in $L''$, meaning the $u'\xrightarrow{\alpha}\bot'$ transition is covered.
        \item If $\alpha=\tau$, then note that $\hat{S}_i\setminus\set{\tau}$ is non-empty, so there is some $\beta\in\hat{S}_i\setminus\set{\tau}$. As shown above, there is some $s'_k$ such that $u'\xrightarrow{\tau}s'_k\xrightarrow{\beta}\bot'$ in $L''$. Then, since $\tau\in\hat{S}_k$, it holds that $u'\xrightarrow{\tau}s'_k\xrightarrow{\tau}\bot'$ in $L''$, showing that the $u'\xrightarrow{\alpha}\bot'$ transition is covered.
    \end{enumerate}
    Hence, each introduced transition is covered in $L''$, so by \Cref{prop:covered} and \Cref{lem:covered_persist}, these may be removed to obtain $L'\setminus\set{(u',\tau,s'_i)}$. This contradicts the minimality of $L'$, and therefore $\mathcal{C}$ is a minimal cover of $U$.
    
    To finalise this proof, we must show that we can compute $\mathcal{C}$ in polynomial time given $L'$. The state $u'$ can be identified as the initial state. The states in $\setof{s'_i}{u'\xrightarrow{\tau}s'_i}$ can be identified by iterating over the outgoing transitions of $u'$, of which there are no more than $2m$. Lastly, from a state $s'_i$ one can determine the corresponding set $S_i$ by iterating over each $S_j$ and picking the set where $s'_i\xrightarrow{\alpha}$ for each $\alpha\in S_j$. This uses the assumption that $S_j\subseteq S_i$ implies $i=j$, and requires no more than $mn$ checks. Thus, $\equiv_x$-minimisation is NP-hard.
\end{proof}

\ldesaturatenew*
\begin{proof}
    We proceed by induction of the number of $\tau$-desaturation steps, $n$. For $n=0$, the lemma holds vacuously, since $L=\bar{L}$ and hence the premise is false.

    For $n>0$, let $\hat{L}$ be the LTS just before the final $\tau$-desaturation step. If $p\xrightarrow{\alpha}q$ in $\hat{L}$, then the result holds inductively. Otherwise, there is some state $p'$ such that $p\xrightarrow{\tau}p'$ and $p'\xrightarrow{\alpha}q$ in $\hat{L}$, and $\bar{L}$ is obtained from $\hat{L}$ by removing the $p\xrightarrow{\tau}p'$. We distinguish two cases, and show there is some state $p''\neq p$ such that $p\xrightarrow{\tau}p''\xRightarrow{\tau}p'$ in $L$, but not $p\xrightarrow{\tau}p''$ in $\bar{L}$.
    \begin{enumerate}
        \item If $p\xrightarrow{\tau}p'$ in $L$, then we may simply take $p''=p$. Note that $p'\neq p$ since $L$ has no inert transitions.
        \item Otherwise, we obtain the result using the induction hypothesis.
    \end{enumerate}
    Likewise, it holds that $p'\xRightarrow{\alpha}q$ in $L$. Therefore, we have $p\xrightarrow{\tau}p''\xRightarrow{\alpha}q$ in $L$, as required.
\end{proof}

\cdesaturatenew*
\begin{proof}
    Since $p_5\xrightarrow{\alpha}q_5$, it also holds that $p_3\xrightarrow{\alpha}q_3$. Therefore, by \Cref{lem:desaturate_new}, there exists some state $p'\neq p$ such that $p_2\xrightarrow{\tau}p'_2\xRightarrow{\alpha}q_2$, but $p_3\nxrightarrow{\tau}p'_3$. We distinguish two cases.
    \begin{enumerate}
        \item If $\alpha=\tau$, then by \Cref{lem:tau_succ}, it holds that $q\sqsubseteq_x p'$. Therefore, $p\xrightarrow{\alpha}q$ is covered by $p\xrightarrow{\tau}p'$ in $L_5\cup\set{(p,\tau,p')}$. Hence, $(p,\alpha,q)\in\Removed_5((p,\tau,p'))$.
        \item If $\alpha\neq\tau$, then since $p'$ in $L_5$ is equivalent to $p'$ in $L_2$, there exists some state $q'\sqsupseteq_x q$ such that $p'_5\xRightarrow{\alpha}q'_5$. By \Cref{lem:tau_succ} and since $\alpha\neq\tau$, we may assume w.l.o.g. that $p'_5\xRightarrow{\tau}p''_5\xrightarrow{\alpha}q'_5$ for some state $p''$. In particular, $p''\sqsubseteq_x p'\sqsubset_x p$ holds by \Cref{lem:tau_succ} and the absence of inert transitions. Therefore, $p\xrightarrow{\alpha}q$ is covered by $p''\xrightarrow{\alpha}q'$ in $L_5\cup\set{(p,\tau,p'')}$. Hence, $(p,\alpha,q)\in\Removed_5((p,\tau,p''))$.
    \end{enumerate}
    In both cases, we achieve our desired result.
\end{proof}

\cover*
\begin{proof}
    Firstly, note that $Y_p\cup Z_p\subseteq\setof{(p,\tau,p')}{p'\sqsubset_x p}$ holds, either by definition of a cover, or by \Cref{lem:tau_succ} and the absence of inert transitions. Thus, it remains to show that for each $(p,\alpha,q)\in\Removable_5(p)$, there is some $(p,\tau,p')\in Y_p\cup Z_p$ such that $(p,\alpha,q)\in\Removed_5((p,\tau,p'))$.

    Let $(p,\alpha,q)\in\Removable_5(p)$. Since each transition of $L_5$ is a transition of $L_3$, it holds that $p_3\xrightarrow{\alpha}q_3$. We distinguish two cases. If $p_2\xrightarrow{\alpha}q_2$, then by definition of $Y_p$, there is some $(p,\tau,p')\in Y_p$ such that $(p,\alpha,q)\in\Removed_5((p,\tau,p'))$. Otherwise, the result holds by \Cref{cor:desaturate_new}.
\end{proof}

\minquotient*
\begin{proof}
    Let $L$ be an LTS, and $\hat{L}\equiv_x L$ be minimal. By \Cref{thm:min}, it holds that $|\hat{L}_7|\leq|\hat{L}|$. Since $\hat{L}$ is minimal, it holds that $|\hat{L}_7|=|\hat{L}|$. By \Cref{thm:can}, $L_5\sim\hat{L}_5$, so let $f$ be an isomorphism from $L_5$ to $\hat{L}_5$. Let $\hat{X}_{f(p)}$ denote the cover computed on \Aref**{line:min_saturate_1} for $f(p)$. Note that from $X_p$, one can compute a cover of the same size for $f(p)$ by considering the set $\setof{(f(p),\tau,f(p'))}{(p,\tau,p')\in X_p}$, and likewise from $\hat{X}_{f(p)}$ one can compute a cover of the same size for $p$. Therefore, since both $X_p$ and $\hat{X}_p$ are minimal, they are of the same size, so there is a bijection, $g$ from $X_p$ to $\hat{X}_p$. Moreover, if $(p,\alpha,q)\in\Removable_5(p)$, then also $(f(p),\alpha,f(q))\in\Removable_5(f(p))$. Therefore, each transition of $L_7$ is either a transition of $L_5$ or belongs to $X_p$, and likewise for $\hat{L}_7$. Hence, we may construct a bijection $h$ from transitions of $L_7$ to transitions of $\hat{L}_7$ as follows:
    \[
    h((p,\alpha,q))\coloneqq\begin{cases}
        (f(p),\alpha,f(q))&\text{if }p_5\xrightarrow{\alpha}q_5\\
        g((p,\alpha,q))&\text{otherwise}\quad.
    \end{cases}
    \]
    Therefore $|L_7|=|\hat{L}_7|=|\hat{L}|$, so $|L_7|$ is minimal.
\end{proof}

\covermax*
\begin{proof}
	Let $p'\sqsubseteq_x p''\sqsubset_x p$, and take $(p,\alpha,q)\in\Removed((p,\tau,p'))$. Let $r'\xrightarrow{\alpha}q'$ be some transition in $L\cup\set{(p,\tau,p')}$ that covers $p\xrightarrow{\alpha}q$. Then there exists some simple path $p_0\xrightarrow{\tau}\dots\xrightarrow{\tau}p_n$ in $L\cup\set{(p,\tau,p')}$ such that $(p,r')=(p_0,p_n)$. We show that $p'\xRightarrow{\alpha}q'$ in $L\cup\set{(p,\tau,p')}$. We distinguish three cases.
	\begin{enumerate}
		\item If $p_0\xrightarrow{\tau}\dots\xrightarrow{\tau}p_n\xrightarrow{\alpha}q'$ holds in $L$, then we have a contradiction since $p\xrightarrow{\alpha}q$ is covered in $L$, but $L$ has no covered transitions.
		\item If $(p,\tau,p')=(p_i,\tau,p_{i+1})$, then $p'\xRightarrow{\tau}r'\xrightarrow{\alpha}q'$, as required.
		\item Otherwise if $(p,\tau,p')=(r',\alpha,q')$, then $\alpha=\tau$ and $p'=q'$, and hence $p'\xRightarrow{\alpha}q'$.
	\end{enumerate}
	Since $p'\sqsubseteq_x p''$ and $\tau$-saturation preserves $\sqsubseteq_x$, there is some $q''\sqsupseteq_x q'$ such that $p''\xRightarrow{\alpha}q''$ in $L$. We show that $(p,\alpha,q)\in\Removed((p,\tau,p''))$ by showing that $p\xrightarrow{\alpha}q$ is covered in $L\cup\set{(p,\tau,p'')}$.
	
	By \Cref{lem:tau_succ}, we may assume w.l.o.g. that $p''\xRightarrow{\tau}r''\xrightarrow{\alpha}q''$ for some state $r''$. By \Cref{lem:tau_succ}, it holds that $r''\sqsubseteq_x p''$. Since $p''\sqsubset_x p$, it holds that $r''\sqsubset_x p$, and hence $(p,q)\neq(r'',q'')$. Therefore, $p\xrightarrow{\alpha}q$ is covered in $L\cup\set{(p,\tau,p'')}$ since it is covered by $r''\xrightarrow{\alpha}q''$.
\end{proof}

\end{document}